\documentclass[10pt,journal,twocolumn,final]{IEEEtran}

\usepackage{eqnarray,amssymb,amsmath,amsthm}
\usepackage{mathrsfs}
\usepackage[utf8]{inputenc}
\usepackage[T1]{fontenc}
\usepackage{graphicx}
\usepackage{epsfig}
\usepackage[english]{babel}
\usepackage{subfigure}
\usepackage{color}
\usepackage{epstopdf}
\usepackage{import}
\usepackage{multirow}
\usepackage{cases}
\usepackage{mathtools}
\usepackage{color, colortbl}
\usepackage{cite}
\usepackage{algorithm}
\usepackage{algpseudocode}
\usepackage{bbm}
 \usepackage{cases}

\newtheorem{lemma}{Lemma}
\newtheorem{definition}{Definition}
\newtheorem{proposition}{Proposition}
\newtheorem{remark}{Remark}

\definecolor{LightCyan}{rgb}{0.85,1,1}

\newcommand{\totalpower}[1]{%
	\overline{P}_{#1}
}



\begin{document}

\title{Fronthaul-Aware Software-Defined Wireless Networks: Resource Allocation and User Scheduling}


\author{
\IEEEauthorblockN{Chen-Feng Liu,~\IEEEmembership{Student Member,~IEEE,} Sumudu Samarakoon,~\IEEEmembership{Student Member,~IEEE,} 
\\ Mehdi Bennis,~\IEEEmembership{Senior Member,~IEEE,} and H. Vincent Poor,~\IEEEmembership{Fellow,~IEEE}}

\thanks{This research was supported in part by TEKES grant 2364/31/2014,  in part by the Academy of Finland project CARMA,  in part by the Nokia Foundation, and in part by the U.S. National Science Foundation under Grant CNS-1456793.
This paper was presented in part at the IEEE Global Communications Conference Workshops, Washington, D.C., USA, in Dec.~2016 \cite{GC16}.}
\thanks{C.-F. Liu, S. Samarakoon, and M. Bennis are with the Centre for Wireless Communications, University of Oulu, Oulu 90014, Finland (e-mail: chen-feng.liu@oulu.fi; sumudu.samarakoon@oulu.fi; mehdi.bennis@oulu.fi).}

\thanks{H. V. Poor is with the Department of Electrical Engineering, Princeton University, NJ 08544, USA (e-mail: poor@princeton.edu).}
}

\maketitle

\begin{abstract}

Software-defined networking (SDN) provides an agile and programmable way to optimize radio access networks via a control-data plane separation. Nevertheless, reaping the benefits of wireless SDN hinges on making optimal use of the limited wireless fronthaul capacity. In this work, the problem of fronthaul-aware resource allocation and user scheduling is studied. To this end, a two-timescale fronthaul-aware SDN control mechanism is proposed in which the controller maximizes the time-averaged network throughput by enforcing a coarse correlated equilibrium in the long timescale.   Subsequently, leveraging the controller's recommendations, each base station schedules its users  using  Lyapunov stochastic optimization  in the short timescale, i.e., at each time slot. Simulation results show that significant network throughput enhancements and up to 40\% latency reduction are achieved with the aid of the SDN controller. Moreover, the gains are more pronounced for denser network deployments.
\end{abstract}


\section{Introduction}\label{Sec: Introduction}
\IEEEPARstart{T}{o} satisfy the ever-increasing capacity enhancement requirements, small cell base stations (BSs) are deployed to boost network capacity and offload traffic \cite{Mehdi/SmallCell,Andrew_smallcell,green_smallcell}.
Nevertheless, due to frequency reuse in adjacent cells, transmissions from neighboring BSs result in severe interference. Although locally-coupled BSs may coordinate their transmissions by exchanging information, this incurs non-negligible signaling and overhead, especially for denser networks, and lacks scalability \cite{Chen:14:Mag}. To address these issues, the concept of software-defined networking (SDN) which decouples the control plane\footnote{The control plane is responsible for radio resource management, admission control, power control, mobility management, scheduling, etc.~\cite{Chen:14:Mag}.} from the data forwarding elements, i.e., the data plane, of the network entity is currently considered as a solution to coordinate the transmissions of locally-coupled radio access networks (RANs) \cite{Chen_mag_15,2015:SOF,Kozat_SDN,Gudipati:13:SDN,Le2015,Yang:2015:SVF}. 
In existing wireless SDN architectures, an SDN controller, having the global view of the hierarchy, makes control decisions and issues recommendations to locally-coupled BSs at the lower-level of the hierarchy. As a result, instead of local coordination with other BSs, BSs optimize their local transmissions while following the controller's recommendations. 

\subsection{Related Work}
The authors in \cite{Fu:2013:NVF} consider a virtualized wireless network in which  wireless service providers (WSPs) bid for resources on behalf of the subscribed users with the aid of the central controller, whereby resources are managed through an auction process. In \cite{SDNAUCTION_PIMRC16}, Zhang {\it et al}.~study virtualized wireless networks in which mobile virtual network operators  aim at acquiring resources from infrastructure providers (InPs) for monetary benefits. Therein, the virtualized resources are allocated by the SDN controller via a double auction mechanism. Considering a large-scale network virtualization architecture in  \cite{Xianfu_ICC_17}, the controller allocates the leased InP resources to users using a distributed Markov decision process algorithm. In \cite{Xianfu:SDN}, the central controller  schedules users based on the WSP's value function which depends on other competitors' private information. Further, the authors propose an algorithm to approximately learn the true value function. Additionally, Liu {\it et al.}~study a traffic-offloading problem in software-defined heterogeneous networks \cite{Liu_SDN_15}. Therein, using a distributed alternating direction method of multipliers  approach, the SDN controller helps BSs to offload their data traffic to access points. The authors in \cite{Zhang_SDN_CRAN} adopt the idea of SDN in cloud radio access networks  where the central controller makes  control decisions for user grouping, relay selection, and resource allocation. In  \cite{Kozat:15:Flow},  the SDN controller splits and allocates network flows into multiple paths to avoid congestion and reduce latency.

However, while interesting, the above works neither take into account the impact of the fronthaul nor  address the overhead induced by using a capacity-limited shared fronthaul which negatively impacts the network performance. Clearly, the impact of the fronthaul capacity and latency on the overall network performance cannot be ignored, and warrants a thorough investigation. This constitutes the prime motivation of this work.

\subsection{Our Contribution}

Motivated by these concerns, we propose a fronthaul-aware software-defined control mechanism for  locally-coupled small cell networks in which  BSs compete for resources to maximize their own downlink (DL) rates while satisfying their user equipments' (UEs') quality-of-service (QoS) requirements with respect to the data queue length. Due to the stochastic nature of the wireless channel, queue dynamics, and  coupled BS transmissions, the rate maximization problem can be modeled as a \emph{stochastic game} among BSs, in which the SDN controller acts as a game coordinator issuing  recommendations to players/BSs via an in-band wireless fronthaul.\footnote{Owing to the advantages of spectrum availability, low expenditure, and flexibility of deployment, using the in-band wireless fronthaul has been selected as a 5G architectural enabler \cite{FH_1,FH_2}.} Here, the fronthaul overhead is considered as a time penalty which is a function of the fronthaul signal-to-noise ratio (SNR) and network density.
In order to reduce the incurred overhead in centralized SDN schemes, we propose a two-timescale framework for the fronthaul-aware software-defined control mechanism in which the controller maximizes the network utility based on BSs' uploaded local information and QoS requirements in a slow/long timescale. Additionally, the controller enforces  a game-theoretic equilibrium, i.e., \emph{coarse correlated equilibrium} (CCE),  to incentivize BSs to attain higher performance.
To this end, by marrying tools from Lyapunov stochastic optimization and game theory, the problem is formulated as a network utility maximization problem subject to CCE constraints. To solve it, we propose two approaches in which the controller provides the optimal CCE recommendations (i.e., allocated frequency carriers) to every BS. In turn, each BS schedules its user as a function of queue length and interference levels. In the fast/short timescale, i.e., at each time slot, the BS utilizes the available frequency carriers to schedule and allocate resources to its UEs by invoking tools from Lyapunov stochastic optimization. The contributions of this paper are summarized as follows.
\begin{itemize}
\item
We propose a two-timescale software-defined control solution for resource allocation and user scheduling.
\item Throughput and latency (in terms of queue length) tradeoffs of the fronthaul-aware wireless SDN framework are examined. 
\item 
In contrast with a non-SDN baseline, the proposed fronthaul-aware software-defined mechanism brings about throughput and latency enhancement reaching up to 40\% latency reduction with the aid of the SDN controller. 
\item
 The impact of the fronthaul overhead and reliability (in terms of  fronthaul SNR) on the performance of software-defined RANs is investigated.
\end{itemize}

The rest of this paper is organized as follows.  We first describe the system model in Section \ref{Sec: System architecture}. The stochastic game of the considered network and the network problem are modeled and formulated in Sections \ref{Sec: Stochastic game} and Section \ref{Sec: Problem formulation}, respectively. Subsequently, we detailedly specify the proposed fronthaul-aware software-defined resource allocation and user scheduling mechanism in Section \ref{Sec: controller's task}.
In Section \ref{Sec: results}, we present  the simulation results. Finally, this work is concluded in Section \ref{Sec: conclusion}.

\section{System Model}\label{Sec: System architecture}

\begin{figure}[t]
\centering
\subfigure[Network architecture.]
{\includegraphics[width=\columnwidth]{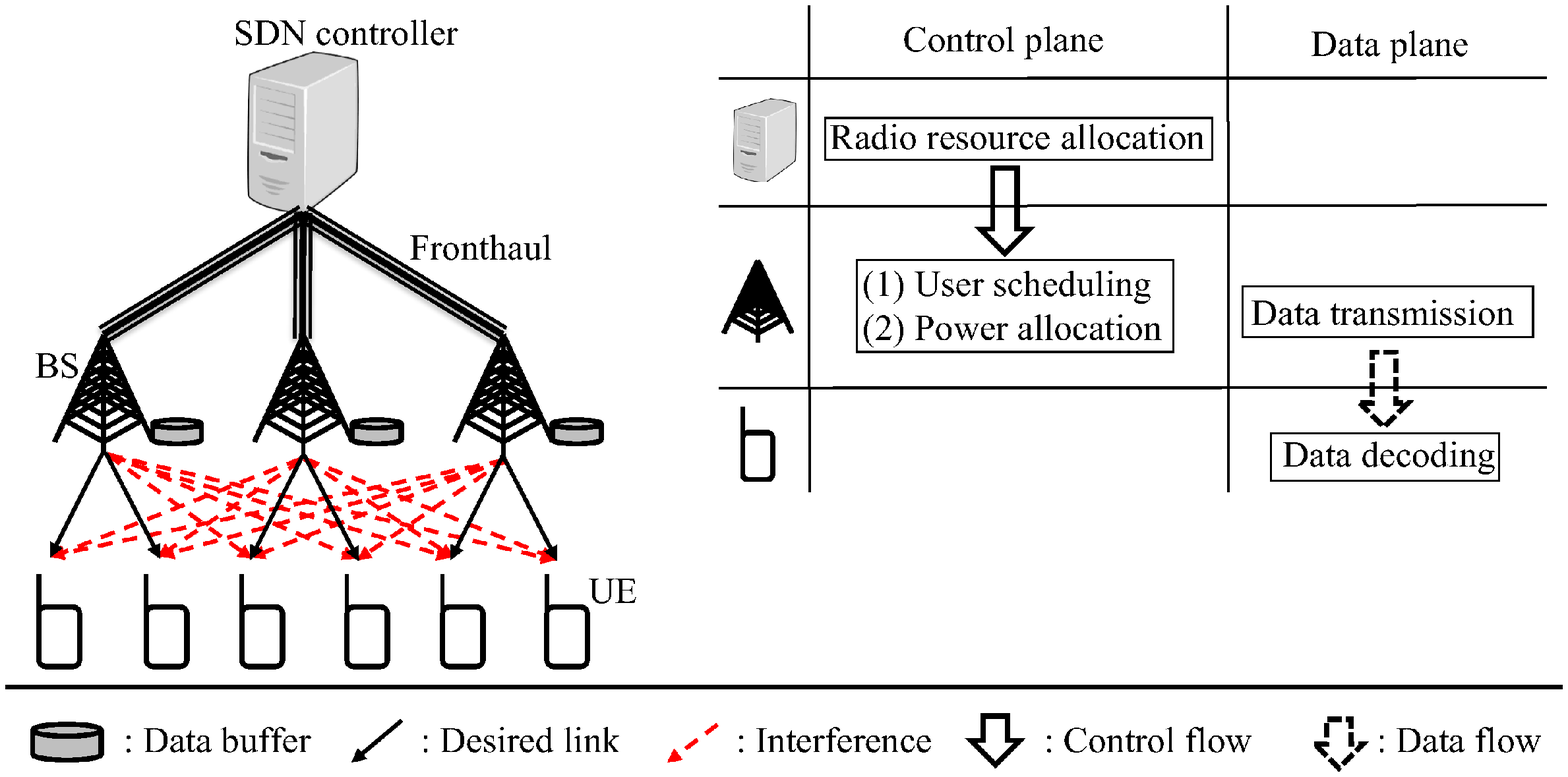}}
\subfigure[Timeline of the proposed two-timescale control mechanism.]
	{\includegraphics[width=\columnwidth]{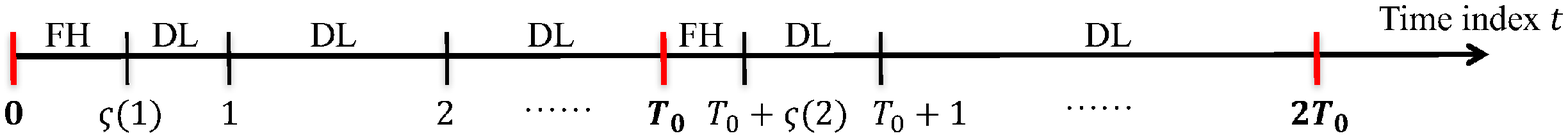}}
	\caption{System model.}
\label{Fig: System}
\end{figure}

As shown in Fig.~\ref{Fig: System}, we consider the DL of a software-defined RAN which consists of a set of locally-coupled small cell BSs $\mathcal{B}$ utilizing the same set of sub-carriers $\mathcal{S}.$ In the considered network, BS $b\in\mathcal{B}$ serves a set of UEs $\mathcal{M}_{b},$ and each UE is served by one BS only, i.e., $\mathcal{M}_{b}\cap\mathcal{M}_{b'}=\emptyset,\forall\,b\neq b'.$  Additionally, an SDN controller, connected to the BSs via an in-band wireless fronthaul, is deployed to coordinate BS transmissions.  We assume that all network entities, i.e., the controller, BSs, and UEs, are equipped with a single antenna. Unless stated otherwise,  we denote the channel gain which includes path loss and channel fading from transmitter $i$ to receiver $j$ over sub-carrier $s\in\mathcal{S}$ as $h_{ij}^{(s)}\in\mathcal{H}_{ij}^{(s)}$ in which $\mathcal{H}_{ij}^{(s)}$ is a finite set. Moreover, the network operates in slotted time indexed by $t\in\mathbb{Z}^{+},$ and $T_0$ successive slots are grouped into one time frame which is indexed by $n\in\mathbb{Z}^{+},$ namely, $\mathcal{T}(n)=[(n-1)T_0+1,\cdots,nT_0].$ For simplicity, each time slot is of unit time. All channels are independent and experience block fading over time slots.

At the beginning of each frame $n,$ BS $b$ uses $\varsigma_{b}^{ \rm U}(n)$ portion of a time slot to send its local information to the controller. To this end, BS $b$ equally allocates the available power budget $|\mathcal{S}|\totalpower b$ over all sub-carriers. We assume that the transmitted information in the fronthaul and DL transmission are encoded based on a Gaussian distribution. Moreover, the noise in the received signal is an additive white Gaussian noise with zero mean and variance $\sigma^2$. Given the required upload rate $R^{\rm U}$ in the fronthaul, we obtain $\varsigma_{b}^{\rm U}(n),\forall\,b\in\mathcal{B},$ from
\begin{equation}\label{Eq: FH time cost}
  \textstyle R^{\mathrm{U}}=  \sum\limits_{s\in\mathcal{S}} \frac{\varsigma_{b}^{\rm U}(n)}{T_0}
\log_2\bigg(1+\frac{\totalpower{b}h^{(s)}_{b{\rm C}}(n)}{\sigma^2+\sum\limits_{b'\in\mathcal{B}\setminus b}\totalpower{b'}h^{(s)}_{b'{\rm C}}(n)}\bigg),
\end{equation}
where the subscript ${\rm C}$ refers to the controller, and $\mathcal{B}\setminus b$ represents the set of BSs excluding BS $b.$ To gather global network information from all BSs, the total required time is represented by $\max_{b\in\mathcal{B}}\{\varsigma_{b}^{\rm U}(n)\}.$
After acquiring the information, the controller uses another time portion $\varsigma_{b}^{\rm F}(n)$ to feed its recommendations back to each BS $b.$ Analogously, 
given the required feedback rate $R^{\rm F},$ the time portion $\varsigma_{b}^{\rm F}(n),$ $\forall\,b\in\mathcal{B},$  can be obtained from 
\begin{align}
 R^{\rm F}&\textstyle=\sum\limits_{s\in\mathcal{S}} \frac{\varsigma_{b}^{\rm F}(n)}{T_0}
\log_2\bigg(1+\frac{\frac{\totalpower  {\rm C}}{|\mathcal{B}|} h^{(s)}_{{\rm C}b}(n)}{\sigma^2+\sum\limits_{b'\in\mathcal{B}\setminus b}\frac{\totalpower {\rm C}}{|\mathcal{B}|}h^{(s)}_{{\rm C}b}(n)}\bigg)\notag
\\&\textstyle=\sum\limits_{s\in\mathcal{S}}\frac{ \varsigma_{b}^{\rm F}(n)}{T_0}
\log_2\Big(1+\frac{\totalpower  {\rm C} h^{(s)}_{{\rm C}b}(n)}{|\mathcal{B}|\sigma^2+(|\mathcal{B}|-1)\totalpower {\rm C}h^{(s)}_{{\rm C}b}(n)}\Big)\label{Eq: FH download cost}
\end{align}
in which the controller's available power $|\mathcal{S}|\totalpower {\rm C}$ is  equally allocated to all BSs over all sub-carriers.\footnote{Extensive studies of other radio access techniques, e.g., multi-input multi-output (MIMO), orthogonal frequency-division multiple access (OFDMA), and non-orthogonal multiple access (NOMA),  in fronthaul-aware transmission are left for future works.} 
Having the controller's recommendations, each BS starts serving its UEs. We note that if one BS starts its DL transmission while other BSs still receive the controller's recommendations via the fronthaul, this incurs interference to the UEs. In order to avoid this, all DL transmissions are synchronized. Therefore, the round-trip time in the fronthaul is given by $\varsigma(n)=\max_{b\in\mathcal{B}}\{\varsigma_{b}^{\rm U}(n)\}+\max_{b\in\mathcal{B}}\{\varsigma_{b}^{\rm F}(n)\}.$ We assume that $\varsigma(n)$ is known at all BSs and belongs to a finite set $\mathcal{G},$ i.e., $\varsigma(n)\in\mathcal{G}.$

In the remaining time of frame $n,$ the BS focuses on DL transmission with the aid of the controller's recommendations. We further assume that each sub-carrier is used by the BS to serve at most one UE, but can be reused by other BSs. Thus, the UE receives inter-cell interference. Since there are $|\mathcal{S}|$ sub-carriers, the BS can orthogonally serve up to $|\mathcal{S}|$ UEs simultaneously.
In each time slot $t,$ BS $b$ allocates power $P_{bm}^{(s)}(t)\in\mathcal{L}_b$ over sub-carrier $s$ to serve UE $m\in\mathcal{M}_{b}.$ Here, $\mathcal{L}_b=\{0,\cdots,|\mathcal{S}|\totalpower b\}$ is a finite set  and $\sum_{m\in\mathcal{M}_b}\sum_{s\in\mathcal{S}} P_{bm}^{(s)}(t)\leq |\mathcal{S}|\totalpower b.$ Taking into account the available time in DL transmission $T_0-\varsigma(n),$ the effective DL rate of UE $m\in\mathcal{M}_b$ served by BS $b\in\mathcal{B}$ over sub-carrier $s\in\mathcal{S}$ and time period $t\in\mathcal{T}(n)$ is expressed as
\begin{equation}\label{Eq: DL rate}
\textstyle R_{bm}^{(s)}(t)=\frac{T_0-\varsigma(n)}{T_0}\log_2\bigg(1+\frac{P_{bm}^{(s)}(t)h_{bm}^{(s)}(t)}{\sigma^2+\sum\limits_{b'\in\mathcal{B}\setminus  b}\sum\limits_{ m'\in\mathcal{M}_{b'}}P_{b'm'}^{(s)}(t)  h_{b'm}^{(s)}(t)}\bigg).
\end{equation}
Note that since the channel gain and transmit power are upper bounded, the DL rate \eqref{Eq: DL rate} is bounded by a maximum value $R_{\rm max}.$

Moreover, BSs have queue buffers to store the data for the served UEs. Denoting the queue length for UE $m\in\mathcal{M}_b$ at the beginning of slot $t$ as $Q_{bm}(t),$ the evolution of queue dynamics is given by
\begin{align}
&\textstyle Q_{bm}(t+1)=\max\Big\{Q_{bm}(t)-\sum\limits_{s\in\mathcal{S}}R_{bm}^{(s)}(t),0\Big\}+\lambda_{bm}(t),\label{Eq: Queue-Q}
\end{align}
where $\lambda_{bm}(t),$ independent and identically distributed over time, is the data arrival during slot $t$ for UE $m\in\mathcal{M}_b.$ In addition, $\lambda_{bm}(t)$ is upper bounded by a finite value $\lambda_{\rm max},$ i.e., $0\leq \lambda_{bm}(t)\leq \lambda_{\rm max},$ and with the mean value $\bar{\lambda}_{bm}>0.$
In practice, it is not feasible to store data indefinitely as time evolves. To cope with this, we force the queue length to be {\it mean rate stable} which means that
\begin{equation*}
\lim\limits_{t\to\infty}\frac{\mathbb{E}\left[|Q_{bm}(t)|\right]}{t}\to 0,~\forall\,b\in\mathcal{B},m\in\mathcal{M}_b.
\end{equation*}
In the considered network, each BS $b$ aims to maximize its long-term time-averaged expected DL rate while stabilizing  UEs' data queues, i.e.,
\begin{IEEEeqnarray*}{rcl}
\mbox{{\bf BP:}}~~&\underset{P_{bm}^{(s)}(t)\in\mathcal{L}_b}{\mbox{maximize}}&\textstyle ~~\sum\limits_{m\in\mathcal{M}_b}\sum\limits_{s\in\mathcal{S}}\bar{R}_{bm}^{(s)}
\\&\mbox{subject to}&~~\lim\limits_{t\to\infty}\frac{\mathbb{E}\left[|Q_{bm}(t)|\right]}{t}\to 0,~\forall\,m\in\mathcal{M}_b,
\\&& \textstyle~~\sum\limits_{m\in\mathcal{M}_b}\mathbbm{1}\big\{ P_{bm}^{(s)}(t)>0\big\}\leq 1,~\forall\,t,s\in\mathcal{S},
\\&&\textstyle ~~\sum\limits_{m\in\mathcal{M}_b}\sum\limits_{s\in\mathcal{S}} P_{bm}^{(s)}(t)\leq |\mathcal{S}|\totalpower b,~\forall\,t,
\end{IEEEeqnarray*}
where $\bar{R}^{(s)}_{bm}\coloneqq\lim\limits_{T\to\infty}\frac{1}{T}\sum_{t=1}^{T}\mathbb{E}\big[R_{bm}^{(s)}(t)\big]$, and the probabilistic expectation $\mathbb{E}[\cdot] $ is calculated with respect to the stochastic channel or/and data arrival. Additionally, $\mathbbm{1}\{\cdot\}$ denotes the indicator function which imposes the constraint that each sub-carrier cannot be used to serve multiple UEs. 
%
%
%
%
%
We also note that \cite{Neely/Stochastic}
\begin{equation*}
\lim\limits_{t\to\infty}\frac{\mathbb{E}\left[|Q_{bm}(t)|\right]}{t}\to 0\Longrightarrow  \textstyle\sum\limits_{s\in\mathcal{S}}\bar{R}^{(s)}_{bm}\geq \bar{\lambda}_{bm}.
\end{equation*}
Therefore, as the BSs compete for sub-carriers to maximize their average DL rates, they need sufficient resources to stabilize the users' queues such that the average DL rate is larger than or equal to the total mean data arrival, i.e.,
\begin{equation}\label{Eq: BS QoS}
\textstyle  \sum\limits_{m\in\mathcal{M}_b} \sum\limits_{s\in\mathcal{S}}\bar{R}^{(s)}_{bm}\geq \bar{\lambda}_b=\sum\limits_{m\in\mathcal{M}_b}\bar{\lambda}_{bm},~\forall\,b\in\mathcal{B},
\end{equation}
where $ \bar{\lambda}_b$ is the sum of all UEs' mean data arrival in $\mathcal{M}_b.$
Since BSs compete for the limited resources to maximize their own utilities (in terms of the average DL rate) under queue dynamics and channel randomness, we model the DL transmission as a stochastic game among BSs, which will be elaborated upon in the next section.

\section{Stochastic Game among Base Stations}\label{Sec: Stochastic game}
The stochastic game among the set of players $\mathcal{B},$ i.e., all BSs, is denoted by $\mathrm{G}=\big(\mathcal{B},\mathcal{W},\mathcal{A},\{u_{b}\}_{b\in\mathcal{B}}\big).$
For the sake of clarity, we list the notation associated with the stochastic game in Table~\ref{Tab: notation}.
\begin{table}[t]
\caption{Notation of the stochastic game.}\label{Tab: notation}
	\vspace{-1em}
\centering
 \begin{tabular}{ll}
\hline
{\bf Definition} &{\bf  Notation } \\
\hline
\hline
BS $b$'s random state &$\boldsymbol{\omega}_b\coloneqq(\varsigma,h_{bm}^{(s)},m\in\mathcal{M}_b,s\in\mathcal{S})$\\
 \hline
BS $b$'s state space&$\mathcal{W}_b\coloneqq\{\boldsymbol{\omega}_b\}$\\
 \hline
Global random state&$\boldsymbol{\omega}\coloneqq[\boldsymbol{\omega}_1,\cdots,\boldsymbol{\omega}_{|\mathcal{B}|}]$\\
\hline
Global  state space&$\mathcal{W}\coloneqq\mathcal{W}_1\times\cdots\times\mathcal{W}_{|\mathcal{B}|}$\\
 \hline
BS $b$'s control action&$\boldsymbol{\alpha}_b\coloneqq\Big(  P^{(s)}_{bm},m\in\mathcal{M}_b,s\in\mathcal{S} \big|\sum\limits_{m\in\mathcal{M}_b}\sum\limits_{s\in\mathcal{S}} $
\\&$P_{bm}^{(s)}\leq |\mathcal{S}|\totalpower b,\sum\limits_{m\in\mathcal{M}_b}\mathbbm{1}\big\{ P_{bm}^{(s)}>0\big\}\leq 1,\Big)$\\
\hline
BS $b$'s action space&$\mathcal{A}_b\coloneqq\big\{\boldsymbol{\alpha}_{b}\big\}$\\
 \hline
Global control action&$\boldsymbol{\alpha}\coloneqq[\boldsymbol{\alpha}_1,\cdots,\boldsymbol{\alpha}_{|\mathcal{B}|}]$\\
\hline
Global action space&$\mathcal{A}\coloneqq\mathcal{A}_1\times\cdots\times\mathcal{A}_{|\mathcal{B}|}$\\
\hline
BS $b$'s utility&$u_{b}(\boldsymbol{\omega},\boldsymbol{\alpha})$\\
\hline
\end{tabular}
\vspace{-2em}
\end {table}
In each time slot $t,$ each BS $b$ first observes a random state $\boldsymbol{\omega}_b(t)\in\mathcal{W}_b,$ and then chooses an action $\boldsymbol{\alpha}_b(t)\in\mathcal{A}_b.$
Given the global random state $\boldsymbol{\omega}(t)\in\mathcal{W}$ and the global action $\boldsymbol{\alpha}(t)\in\mathcal{A},$ BS $b$'s utility is the expected DL rate with respect to the interference channels of $\mathcal{M}_b,$ i.e.,
\begin{align}\label{Eq: BS's utility}
 \textstyle u_{b}(\boldsymbol{\omega}(t),\boldsymbol{\alpha}(t))\coloneqq\sum\limits_{m\in\mathcal{M}_b}\sum\limits_{s\in\mathcal{S}}\mathbb{E}_{\boldsymbol{\iota}_b}\big[R_{bm}^{(s)}(t)\big],
\end{align}
where $\boldsymbol{\iota}_b \coloneqq (h_{b'm}^{(s)},b'\in\mathcal{B}\setminus b,m\in\mathcal{M}_b).$ Further, given the \emph{mixed Markovian} strategy $\Pr(\boldsymbol{\alpha}(t)|  \boldsymbol{\omega}(t))$ of the stochastic game in time slot $t$, i.e., the probability of choosing action $\boldsymbol{\alpha}(t)$ for a given state $\boldsymbol{\omega}(t),$ BS $b$'s expected utility (over the random state distribution and mixed Markovian strategy) in slot $t$ is calculated as 
\begin{multline}
\textstyle\mathbb{E}[ u_{b}(t)]=\sum\limits_{\boldsymbol{\omega}(t)\in\mathcal{W}}\sum\limits_{\boldsymbol{\alpha}(t)\in\mathcal{A}}\Pr(\boldsymbol{\omega}(t))
\\\times \Pr(\boldsymbol{\alpha}(t)|  \boldsymbol{\omega}(t))u_b(\boldsymbol{\omega}(t),\boldsymbol{\alpha}(t)),
 \end{multline}
 and the long-term time-averaged expected utility is expressed as
\begin{multline}
 \bar{u}_b= \textstyle\lim\limits_{T\to\infty}\frac{1}{T} \sum\limits_{t=1}^T
\mathbb{E}[ u_{b}(t)] 
 = \lim\limits_{T\to\infty}\frac{1}{T} \sum\limits_{t=1}^T\sum\limits_{\boldsymbol{\omega}(t)\in\mathcal{W}}
\\\textstyle\sum\limits_{\boldsymbol{\alpha}(t)\in\mathcal{A}} \Pr(\boldsymbol{\omega}(t))\Pr(\boldsymbol{\alpha}(t)|  \boldsymbol{\omega}(t))u_b(\boldsymbol{\omega}(t),\boldsymbol{\alpha}(t)).\label{Eq: BS's time-averaged utility}
\end{multline}
In this work, we focus on the \emph{mixed stationary} and \emph{Markovian} strategy $\Pr(\boldsymbol{\alpha}|\boldsymbol{\omega})$, i.e., $\Pr(\boldsymbol{\alpha}(t)|  \boldsymbol{\omega}(t))=\Pr(\boldsymbol{\alpha}|\boldsymbol{\omega}),\forall\,t.$ Since the random state distribution is also stationary (due to block fading), \eqref{Eq: BS's time-averaged utility} can be rewritten as
\begin{equation}\label{Eq: BS's stationary utility}
 \textstyle\bar{u}_{b}=\sum\limits_{\boldsymbol{\omega}\in\mathcal{W}}\sum\limits_{\boldsymbol{\alpha}\in\mathcal{A}} \Pr(\boldsymbol{\omega})\Pr(\boldsymbol{\alpha}|  \boldsymbol{\omega})u_b(\boldsymbol{\omega},\boldsymbol{\alpha}).
\end{equation}

In this  stochastic game, the controller acts as a game coordinator which issues transmission recommendations, e.g., a mixed strategy $\Pr(\boldsymbol{\alpha}|\boldsymbol{\omega}),$ to the BSs based on the obtained information and the BSs' requirements \eqref{Eq: BS QoS}. To incentivize the BSs to follow the controller's recommendations as introduced in Section \ref{Sec: Introduction}, the controller enforces the CCE strategy. In this regard, we consider the general formulation of $\epsilon$-coarse correlated equilibrium ($\epsilon$-CCE) \cite{Neely:13:Arxiv,epsilonCCE}, formally defined as follows.
\begin{definition}
The  strategy $\Pr(\boldsymbol{\alpha}|\boldsymbol{\omega})$ is an $\epsilon$-CCE of the stochastic game $\mathcal{G}$ if it satisfies
\begin{align}
\Pr(\boldsymbol{\omega}_b)\theta_b( \boldsymbol{\omega}_b)&\textstyle\geq  \sum\limits_{\boldsymbol{\omega}\in\mathcal{W}|\boldsymbol{\omega}_b}\sum\limits_{\boldsymbol{\alpha}\in\mathcal{A}} \Pr(\boldsymbol{\omega})\Pr(\boldsymbol{\alpha}|  \boldsymbol{\omega})u_b(\boldsymbol{\omega}, \boldsymbol{\chi}_b, \boldsymbol{\alpha}_{-b}), \notag
\\&\textstyle\hspace{5.1em}\forall\,b\in\mathcal{B},\boldsymbol{\omega}_b\in\mathcal{W}_b, \boldsymbol{\chi}_b\in\mathcal{A}_b,
\label{Eq: epsilon-CCE-1}
\\  \bar{u}_b&\textstyle\geq 
\sum\limits_{ \boldsymbol{\omega}_b\in\mathcal{W}_b} \Pr( \boldsymbol{\omega}_b)\theta_b( \boldsymbol{\omega}_b)-\epsilon,~\forall\,b\in\mathcal{B},
\label{Eq: epsilon-CCE-2}
\end{align}
with $\epsilon\geq 0.$ 
\end{definition}
In \eqref{Eq: epsilon-CCE-1}, $\boldsymbol{\alpha}_{-b}$ denotes the global action vector excluding BS $b$'s action, and $\boldsymbol{\chi}_b\coloneqq\big(  \chi^{(s)}_{bm},m\in\mathcal{M}_b,s\in\mathcal{S}  \big)$ is one specific action of BS $b$ regardless of the strategy $\Pr(\boldsymbol{\alpha}|  \boldsymbol{\omega})$.
Hence, given the observed state $\boldsymbol{\omega}_b$ of BS $b,$ $\theta_b( \boldsymbol{\omega}_b)\in[0,u_{b}^{\rm max}]$ is the maximum utility if BS $b$  deviates from the $\epsilon$-CCE strategy $\Pr(\boldsymbol{\alpha}|  \boldsymbol{\omega})$ whereas the other BSs follow it. $u_{b}^{\rm max}$ is the maximum utility BS $b$ can achieve with any strategy.
\begin{remark}
Given that the other BSs follow the controller's $\epsilon$-CCE recommendations, BS $b$'s utility gain is upper bounded by $\epsilon$ if it deviates. Moreover, by letting $\epsilon=0,$ we have the expression for the conventional CCE \cite{jnl:perlaza13}.
\end{remark}

\section{Network Utility Maximization Formulation}\label{Sec: Problem formulation}

Taking into account the BSs' QoS requirements \eqref{Eq: BS QoS} and $\epsilon$-CCE constraints  \eqref{Eq: epsilon-CCE-1} and \eqref{Eq: epsilon-CCE-2}, the controller aims at solving the following  problem:
\begin{IEEEeqnarray*}{rcl}
\mbox{{\bf MP:}}~~&\underset{\Pr(\boldsymbol{\alpha}|\boldsymbol{\omega}),\theta_b( \boldsymbol{\omega}_b)}{\text{maximize}} &~~\phi(\bar{u}_{1},\cdots,\bar{u}_{|\mathcal{B}|})
\\&\mbox{subject to}&~~\mbox{QoS requirement \eqref{Eq: BS QoS}},
\\&&~~\mbox{$\epsilon$-CCE constraints \eqref{Eq: epsilon-CCE-1} and \eqref{Eq: epsilon-CCE-2}},
\\&&\textstyle~~ \sum\limits_{\boldsymbol{\alpha}\in\mathcal{A}} \Pr(\boldsymbol{\alpha}|  \boldsymbol{\omega})=1,~\forall\,\boldsymbol{\omega}\in\mathcal{W},
\\&&~~\Pr(\boldsymbol{\alpha}|  \boldsymbol{\omega})\geq 0,~\forall\,\boldsymbol{\omega}\in\mathcal{W}, \boldsymbol{\alpha}\in\mathcal{A},
\\&&~~\theta_b( \boldsymbol{\omega}_b)\in[0,u_{b}^{\rm max}],~\forall\,b\in\mathcal{B},\boldsymbol{\omega}_b\in\mathcal{W}_b,
\end{IEEEeqnarray*}
with the network utility function
\begin{align}
&\textstyle \phi(\bar{u}_{1},\cdots,\bar{u}_{|\mathcal{B}|})\coloneqq\sum\limits_{b\in\mathcal{B}}\bar{\lambda}_b\ln\big(1+\bar{u}_b\big)\label{Eq: Network utility}
 \\&\textstyle=\sum\limits_{b\in\mathcal{B}}\bar{\lambda}_b\ln\Big(1+\sum\limits_{\boldsymbol{\omega}\in\mathcal{W}}\sum\limits_{\boldsymbol{\alpha}\in\mathcal{A}} \Pr(\boldsymbol{\omega})\Pr(\boldsymbol{\alpha}|  \boldsymbol{\omega})u_b(\boldsymbol{\omega},\boldsymbol{\alpha})\Big).\label{Eq: Network utility_expansion}
\end{align}
The network utility  has a finite value and ensures that the BSs' average DL rates are proportional to their data arrival rates. To solve {\bf MP}, the controller requires the statistics of the interference channels $\boldsymbol{\iota}_b$ in $u_{b}(\boldsymbol{\omega},\boldsymbol{\alpha})$ as per \eqref{Eq: DL rate}, \eqref{Eq: BS's utility}, and \eqref{Eq: BS's stationary utility}. Since the UE  measures only the aggregate interference from the received signals, having the distribution of each interference channel, even approximately, is highly complex. To address this issue, we introduce an auxiliary utility for each $b\in\mathcal{B},$ i.e.,
\begin{align}
&\textstyle v_{b}(\boldsymbol{\omega},\boldsymbol{\alpha})\coloneqq  \sum\limits_{m\in\mathcal{M}_b}\sum\limits_{s\in\mathcal{S}}\frac{T_0-\varsigma}{T_0}\notag
\\&\textstyle\qquad\times \log_2\bigg(1+\frac{P_{bm}^{(s)}h_{bm}^{(s)}}{\sigma^2+\sum\limits_{b'\in\mathcal{B}\setminus b}\sum\limits_{ m'\in\mathcal{M}_{b'}}P_{b'm'}^{(s)}  [\mathcal{H}_{b'm}^{(s)}]_{\rm max}} \bigg),\label{Eq: Auxiliary utility}
\end{align}
where $ [\mathcal{H}_{b'm}^{(s)}]_{\rm max}$ is the maximum element in  $\mathcal{H}_{b'm}^{(s)}.$
\begin{proposition}\label{Prop: Mean rate stability}
The QoS requirement \eqref{Eq: BS QoS} is satisfied  if
\begin{equation}\label{Eq: Auxiliary mean stable}
 \textstyle\bar{v}_b= \sum\limits_{\boldsymbol{\omega}\in\mathcal{W}}\sum\limits_{\boldsymbol{\alpha}\in\mathcal{A}}  \Pr(\boldsymbol{\omega})\Pr(\boldsymbol{\alpha}|  \boldsymbol{\omega})v_{b}(\boldsymbol{\omega},\boldsymbol{\alpha})\geq\bar{\lambda}_{b}.
\end{equation}
\end{proposition}
\begin{proof}
Please refer to Appendix \ref{Lem: Mean rate stability}.
\end{proof}
\begin{proposition}\label{Prop: epsilon-CCE}
The CCE strategy with respect to  the auxiliary utility $v_b(\boldsymbol{\omega},\boldsymbol{\alpha})$ achieves an $\epsilon$-CCE with respect to the  utility $u_b(\boldsymbol{\omega},\boldsymbol{\alpha}).$
\end{proposition}
\begin{proof}
Please refer to Appendix \ref{Lem: epsilon-CCE}.
\end{proof}
Considering the results of Propositions \ref{Prop: Mean rate stability} and \ref{Prop: epsilon-CCE},  we reformulate the network problem {\bf MP} as
\begin{subequations}\label{Eq: Auxiliary network problem}
\begin{IEEEeqnarray}{cl}
\hspace{-2.5em}\underset{\Pr(\boldsymbol{\alpha}|\boldsymbol{\omega}),\theta_b( \boldsymbol{\omega}_b)}{\text{maximize}} &~\phi(\bar{v}_{1},\cdots,\bar{v}_{|\mathcal{B}|})\label{Eq: Auxiliary network problem-1}
\\\hspace{-2.5em}\mbox{subject to}&~\mbox{QoS requirement \eqref{Eq: Auxiliary mean stable}},\notag
\\&~\mbox{CCE constraints \eqref{Eq: Auxiliary CCE-1} and \eqref{Eq: Auxiliary CCE-2}},\notag
\\&~\textstyle \sum\limits_{\boldsymbol{\alpha}\in\mathcal{A}} {\Pr}(\boldsymbol{\alpha}|  \boldsymbol{\omega})=1,~\forall\,\boldsymbol{\omega}\in\mathcal{W},\label{Eq: Auxiliary network problem-4}
\\&~{\Pr}(\boldsymbol{\alpha}|  \boldsymbol{\omega})\geq 0,~\forall\,\boldsymbol{\omega}\in\mathcal{W}, \boldsymbol{\alpha}\in\mathcal{A},\label{Eq: Auxiliary network problem-5}
\\&~\theta_b(\boldsymbol{\omega}_b)\in[0,v_{b}^{\rm max}],~\forall\,b\in\mathcal{B},\boldsymbol{\omega}_b\in\mathcal{W}_b.
\end{IEEEeqnarray}
\end{subequations}
In this reformulated problem, we consider $\bar{v}_b$ in the network utility function, and $v_{b}^{\rm max}$ is the maximum achievable utility given any strategy. Note that the feasible set of the problem \eqref{Eq: Auxiliary network problem} is a subset of the feasible set of {\bf MP}, and $\bar{v}_b$ is upper bounded by $\bar{u}_b.$ Therefore, the optimal solution to problem \eqref{Eq: Auxiliary network problem} is a lower bound on the optimal solution to {\bf MP}.

\section{Software-defined Resource Allocation and User Scheduling}\label{Sec: controller's task}

In the proposed fronthaul-aware software-defined control mechanism, the controller issues recommendations based on the optimal CCE strategy $\Pr^{*}(\boldsymbol{\alpha}|\boldsymbol{\omega})$ of \eqref{Eq: Auxiliary network problem}. Nevertheless,  directly solving \eqref{Eq: Auxiliary network problem} suffers from the curse of computational complexity and incurs higher fronthaul cost. To alleviate these issues,  we resort to Lyapunov optimization which will be explained next.

\subsection{Statistics-based Resource Allocation  at the Controller}\label{Sec: Empirical information}

In order to solve problem \eqref{Eq: Auxiliary network problem}, the controller requires the knowledge of the global state statistics and mean data arrivals. Although the exact statistical information is unknown beforehand, the controller can acquire the empirically estimated information from the BSs via the fronthaul. Let us denote the estimated global sate statistics and the mean data arrival at BS $b$ as 
\begin{align}
\textstyle\hat{\Pr}(\boldsymbol{\omega}) =\prod\limits_{b\in\mathcal{B}}\prod\limits_{m\in\mathcal{M}_b}\prod\limits_{s\in\mathcal{S}}\hat{\Pr}(\varsigma)\hat{\Pr}(h_{bm}^{(s)})\label{Eq: estimated channel dist}
\end{align}
and $\hat{\lambda}_b,$ respectively. Leveraging the estimated information, \eqref{Eq: Auxiliary network problem} can be rewritten as
\begin{subequations}\label{Eq: Estimated auxiliary network problem}
\begin{IEEEeqnarray}{cl}
\hspace{-0.5em}\underset{\Pr(\boldsymbol{\alpha}|\boldsymbol{\omega}),\theta_b(\boldsymbol{\omega}_b)}{\mbox{maximize}} &\textstyle \sum\limits_{b\in\mathcal{B}}\hat{\lambda}_b\ln(1+\hat{v}_b)\label{Eq: Learning-objecitve}
\\\hspace{-0.5em}\mbox{subject to}&\textstyle\hat{v}_b\geq\hat{\lambda}_b,~\forall\,b\in\mathcal{B},\label{Eq: Learning-QoS}
\\\hspace{-0.5em}&\textstyle \hat{\Pr}(\boldsymbol{\omega}_b) \theta_b( \boldsymbol{\omega}_b)
\geq \sum\limits_{\boldsymbol{\omega}\in\mathcal{W}|\boldsymbol{\omega}_b\atop\boldsymbol{\alpha}\in\mathcal{A}} v_{b}(\boldsymbol{\omega},\boldsymbol{\chi}_{b},\boldsymbol{\alpha}_{-b}) \hat{\Pr}(\boldsymbol{\omega}) \notag
\\\hspace{-0.5em}&\times{\Pr}(\boldsymbol{\alpha}|  \boldsymbol{\omega}),~\forall\,b\in\mathcal{B},\boldsymbol{\omega}_b\in\mathcal{W}_b, \boldsymbol{\chi}_b\in\mathcal{A}_b, \label{Eq: Learning-CCE-1}
\\\hspace{-0.5em}&\textstyle \hat{v}_b\geq\sum\limits_{\boldsymbol{\omega}_b\in\mathcal{W}_b} \hat{\Pr}(\boldsymbol{\omega}_b)\theta_b( \boldsymbol{\omega}_b),~\forall\,b\in\mathcal{B},\label{Eq: Learning-CCE-2}
\\\hspace{-0.5em}&\textstyle \sum\limits_{\boldsymbol{\alpha}\in\mathcal{A}} \Pr(\boldsymbol{\alpha}|  \boldsymbol{\omega})=1,~\forall\,\boldsymbol{\omega}\in\mathcal{W},\label{Eq: Learning-dist-1}
\\\hspace{-0.5em}&\textstyle\Pr(\boldsymbol{\alpha}|  \boldsymbol{\omega})\geq 0,~\forall\,\boldsymbol{\omega}\in\mathcal{W}, \boldsymbol{\alpha}\in\mathcal{A},\label{Eq: Learning-dist-2}
\\\hspace{-0.5em}&{\theta}_b(\boldsymbol{\omega}_b)\in[0,{v}_{b}^{\rm max}],~\forall\,b\in\mathcal{B},\boldsymbol{\omega}_b\in\mathcal{W}_b.\label{Eq: Learning-dist-3}
\end{IEEEeqnarray}
\end{subequations}
with $\hat{v}_b=\sum_{\boldsymbol{\omega}\in\mathcal{W}}\sum_{\boldsymbol{\alpha}\in\mathcal{A}}  \hat{\Pr}(\boldsymbol{\omega}){\Pr}(\boldsymbol{\alpha}|  \boldsymbol{\omega})v_{b}(\boldsymbol{\omega},\boldsymbol{\alpha})$.
Since the objective and all constraints are concave and affine functions, respectively, problem \eqref{Eq: Estimated auxiliary network problem} is a convex optimization problem. However,  finding a closed-form solution is not tractable even after applying the Karush-Kuhn-Tucker (KKT) conditions.
To tackle this issue, we resort to CVX, a computer-based package for convex optimization \cite{cvx}, to numerically calculate the optimal solution ${\Pr}^{*}(\boldsymbol{\alpha}|\boldsymbol{\omega}).$

Note that the fronthaul cost depends on the amount of feedback and uploaded information.
If the controller directly feeds back the optimal strategy ${\Pr}^{*}(\boldsymbol{\alpha}|\boldsymbol{\omega})$ to BSs, there are $|\mathcal{A}|\times|\mathcal{W}|$ probabilistic values to be sent which increase exponentially with the number of BSs, UEs, and sub-carriers. This imposes a heavy burden on the fronthaul when the number of network entities grows. To alleviate this overhead while feeding back transmission recommendations, we consider the following scenario. According to the optimal CCE strategy ${\Pr}^{*}(\boldsymbol{\alpha}|\boldsymbol{\omega}),$ the controller generates a set of  $T_0$ global action realizations, i.e.,  $\{\boldsymbol{\alpha}_{\rm I}((n-1)T_0+1),\cdots,\boldsymbol{\alpha}_{\rm I}(nT_0)\},$ for each global state $\boldsymbol{\omega}\in\mathcal{W}.$
Then from all global states and the corresponding generated action realizations $\boldsymbol{\alpha}_{\rm I}(t)\in\mathcal{A}$ for time slot $t,$ the controller finds the mapping rule $\Psi^{*}_{\rm I}(t):\mathcal{W}\to\mathcal{A}$ which describes the above relation between  states and action realizations, i.e., $\Psi^{*}_{\rm I}(\boldsymbol{\omega};t)=\boldsymbol{\alpha}_{\rm I}(t).$ We assume that all possible mapping rules are known at the controller and BSs beforehand, and each possible $\Psi_{\rm I}$ is represented by a real value. Finally, the controller feeds back $T_0$ real values for $\{\Psi^{*}_{\rm I}(t),t\in\mathcal{T}(n)\}$ to the BSs as control recommendations. 
In the upload phase, the amount of uploaded information is decided by the mean data arrival as well as the cardinalities of the round trip time in fronthaul and local state space. Thus, there are $1+|\mathcal{G}|+\sum_{m\in\mathcal{M}_b,s\in\mathcal{S}}|\mathcal{H}_{bm}^{(s)}|$ uploaded statistical values for BS $b.$
Given that $R_{\rm unit}$ is the required transmission rate to upload a single statistical value, we can substitute  $(1+|\mathcal{G}|+\sum_{m\in\mathcal{M}_b,s\in\mathcal{S}}|\mathcal{H}_{bm}^{(s)}|)R_{\rm unit}$ for $R^{\rm U}$ in \eqref{Eq: FH time cost}. Similarly, we have $T_0R_{\rm unit}$ for $R^{\rm F}$  in \eqref{Eq: FH download cost}.
Regarding the computational complexity of the optimization problem \eqref{Eq: Estimated auxiliary network problem}, there are 
\begin{multline}\notag
\textstyle|\mathcal{W}| |\mathcal{A}|+\sum\limits_{b\in\mathcal{B}}|\mathcal{W}_b|
\approx
|\mathcal{G}| \Big(\prod\limits_{b\in\mathcal{B}}\prod\limits_{m\in\mathcal{M}_b}\prod\limits_{s\in\mathcal{S}}|\mathcal{H}_{bm}^{(s)}||\mathcal{L}_b|\Big)
\\\textstyle+\sum\limits_{b\in\mathcal{B}}|\mathcal{G}|\Big(\prod\limits_{m\in\mathcal{M}_b}\prod\limits_{s\in\mathcal{S}}|\mathcal{H}_{bm}^{(s)}|\Big)
\end{multline}
 optimization variables and
\begin{align*}
&\textstyle 2|\mathcal{B}|+\sum\limits_{b\in\mathcal{B}}|\mathcal{W}_b| |\mathcal{A}_b|+|\mathcal{W}|+|\mathcal{W}| |\mathcal{A}|+2\sum\limits_{b\in\mathcal{B}}|\mathcal{W}_b| 
\\&\textstyle\approx 2|\mathcal{B}|+\sum\limits_{b\in\mathcal{B}}|\mathcal{G}| \Big(\prod\limits_{m\in\mathcal{M}_b}\prod\limits_{s\in\mathcal{S}}|\mathcal{H}_{bm}^{(s)}||\mathcal{L}_b|\Big)+|\mathcal{G}| \Big(\prod\limits_{b\in\mathcal{B}}\prod\limits_{m\in\mathcal{M}_b}
\\&\quad\textstyle\prod\limits_{s\in\mathcal{S}}|\mathcal{H}_{bm}^{(s)}|\Big)
+|\mathcal{G}| \Big(\prod\limits_{b\in\mathcal{B}}\prod\limits_{m\in\mathcal{M}_b}\prod\limits_{s\in\mathcal{S}}|\mathcal{H}_{bm}^{(s)}||\mathcal{L}_b|\Big)
\\&\quad\textstyle+2\sum\limits_{b\in\mathcal{B}}|\mathcal{G}|\Big(\prod\limits_{m\in\mathcal{M}_b}\prod\limits_{s\in\mathcal{S}}|\mathcal{H}_{bm}^{(s)}|\Big)
\end{align*}
 affine constraints. Both increase exponentially with the number of BSs, UEs, and sub-carriers.

We briefly summarize the information exchange in the fronthaul of the statistics-based resource allocation approach. In the first time portion $\varsigma(n)$ of each frame $n$, BSs first upload \eqref{Eq: estimated channel dist} and the empirically estimated mean data arrivals to the controller, then the controller feeds back $\{\Psi^{*}_{\rm I}(t),t\in\mathcal{T}(n)\}$ to BSs.

\subsection{State Realization-based Resource Allocation  at the Controller}\label{Sec: Lyapunov optimization}

The main drawback of the former solution is that the number of optimization variables and constraints increases exponentially with the increasing number of BSs, UEs, and sub-carriers. Thus, as the network becomes dense,
 the controller suffers from high computational complexity, and uploading local information imposes a heavy burden on the fronthaul. To remedy to this, we employ tools from  Lyapunov optimization, a stochastic optimization approach with low computational complexity that requires only the current state realization instead of global state statistics.

\subsubsection{Lyapunov Optimization Framework}
To solve \eqref{Eq: Estimated auxiliary network problem}, we first introduce a continuous random variable $\gamma_b$ for each BS $b\in\mathcal{B}$ with the probability density function $f_{\gamma_b}$ over the domain $[0, v_{b}^{\rm max}].$
In addition, referring to \eqref{Eq: Network utility}, we define an auxiliary function 
\begin{equation}\label{Eq: Auxiliary network utility}
\textstyle\sum\limits_{b\in\mathcal{B}}\hat{\lambda}_b\ln\big(1+\gamma_b\big),
\end{equation}
for the network utility \eqref{Eq: Learning-objecitve}. With the aid of the auxiliary random variables and function,  \eqref{Eq: Estimated auxiliary network problem}  is equivalent to
\begin{subequations}\label{Eq: Lyapunov network problem}
\begin{IEEEeqnarray}{cl}
\hspace{-2em}\underset{\Pr(\boldsymbol{\alpha}|\boldsymbol{\omega}),\theta_b( \boldsymbol{\omega}_b),f_{\gamma_b}}{\text{maximize}}&\textstyle~~\sum\limits_{b\in\mathcal{B}}\hat{\lambda}_b\cdot\mathbb{E}\big[\ln(1+\gamma_b)\big]\label{Eq: Lyapunov network problem-1}
\\\hspace{-2em}\mbox{subject to}&\textstyle~~\mathbb{E}[\gamma_b]=\sum\limits_{\boldsymbol{\omega}\in\mathcal{W}}\sum\limits_{\boldsymbol{\alpha}\in\mathcal{A}}  \hat{\Pr}(\boldsymbol{\omega})\notag
\\\hspace{-2em}&~~\qquad\times\Pr(\boldsymbol{\alpha}|  \boldsymbol{\omega})v_{b}(\boldsymbol{\omega},\boldsymbol{\alpha}),~\forall\,b\in\mathcal{B},\label{Eq: Lyapunov network problem-2}
\\\hspace{-2em}&~~\mbox{QoS requirement \eqref{Eq: Learning-QoS}},\notag
\\\hspace{-2em}&~~\mbox{CCE constraints \eqref{Eq: Learning-CCE-1} and \eqref{Eq: Learning-CCE-2}},\notag
\\\hspace{-2em}&~~\mbox{\eqref{Eq: Learning-dist-1}, \eqref{Eq: Learning-dist-2}, and \eqref{Eq: Learning-dist-3}}.\notag
\end{IEEEeqnarray}
\end{subequations}
%
%
Subsequently, we  introduce virtual queues for each of the constraints \eqref{Eq: Learning-QoS}, \eqref{Eq: Learning-CCE-1}, \eqref{Eq: Learning-CCE-2}, and \eqref{Eq: Lyapunov network problem-2}.
For the CCE constraints \eqref{Eq: Learning-CCE-1} and \eqref{Eq: Learning-CCE-2}, we have virtual queues $Y^{\boldsymbol{\chi}_b}_{\boldsymbol{\omega}_b}$ and $Z_{b},$ respectively, which evolve over time slots as follows:
\begin{align}
&Y^{\boldsymbol{\chi}_b}_{\boldsymbol{\omega}_b}(t+1)=\max\big\{Y^{\boldsymbol{\chi}_b}_{\boldsymbol{\omega}_b}(t)+\tilde{v}_{b}(\boldsymbol{\omega}_b,\boldsymbol{\chi}_b;t)-\tilde{\theta}_b(\boldsymbol{\omega}_b;t),0\big\},\notag
\\&\hspace{10.5em}~\forall\,b\in\mathcal{B},\boldsymbol{\omega}_b\in\mathcal{W}_b, \boldsymbol{\chi}_b \in\mathcal{A}_b,\label{Eq: VQ-Y}
\\&\textstyle  Z_{b}(t+1)=\max\Big\{Z_b(t)+\sum\limits_{\boldsymbol{\omega}_b\in\mathcal{W}_b}\tilde{\theta}_b(\boldsymbol{\omega}_b;t)\notag
\\&\hspace{9.5em}-v_{b}(\mathbf{h}(t),\boldsymbol{\alpha}(t)),0\big\},~\forall\,b\in\mathcal{B},\label{Eq: VQ-Z}
\end{align}
where
\begin{align}
\tilde{\theta}_b(\boldsymbol{\omega}_b;t)=
\begin{cases}
\tilde{\theta}_b(\boldsymbol{\omega}_b;t)\in[0, v_{b,\rm max}],&\mbox{if }\boldsymbol{\omega}_b=\mathbf{h}_b(t),
\\0,&\mbox{otherwise},
\end{cases}
\label{Eq: Auxiliary function-1}
\end{align}
and
\begin{align}
\tilde{v}_{b}(\boldsymbol{\omega}_{b},\boldsymbol{\chi}_{b};t)=
\begin{cases}v_{b}(\boldsymbol{\omega}_b,\mathbf{h}_{-b}(t),\boldsymbol{\chi}_b,\boldsymbol{\alpha}_{-b}(t)),&\mbox{if }
\boldsymbol{\omega}_b=\mathbf{h}_b(t),
\\0,&\mbox{otherwise,}
\end{cases}
\label{Eq: Auxiliary function-2}
\end{align}
are the auxiliary functions corresponding to $\theta_b(\boldsymbol{\omega}_b)$ and $v_b(\boldsymbol{\omega},\boldsymbol{\chi}_{b},\boldsymbol{\alpha}_{-b})$ in \eqref{Eq: Learning-CCE-1}, respectively.
$\mathbf{h}(t)=[\mathbf{h}_1(t),\cdots,\mathbf{h}_{|\mathcal{B}|}(t)]\in\mathcal{W}$ is the global state realization in time slot $t$.
In addition, virtual queues
\begin{align}
\hspace{-0.3em}D_{b}(t+1)&=\max\big\{D_{b}(t)+\hat{\lambda}_b-v_{b}(\mathbf{h}(t),\boldsymbol{\alpha}(t)),0\big\},~\forall\,b\in\mathcal{B},\label{Eq: VQ-D}
\\\hspace{-0.3em}F_{b}(t+1)&=F_{b}(t)+\gamma_{b}(t)-v_{b}(\mathbf{h}(t),\boldsymbol{\alpha}(t)),~\forall\,b\in\mathcal{B},\label{Eq: VQ-F}
\end{align}
are introduced for \eqref{Eq: Learning-QoS} and \eqref{Eq: Lyapunov network problem-2}, respectively, with $0\leq \gamma_b(t)\leq v_{b}^{\rm max}.$ Using Lyapunov optimization, the optimal solution of \eqref{Eq: Lyapunov network problem-1} is obtained by optimizing the equivalent long-term time-averaged objective $\lim\limits_{T\to\infty}\frac{1}{T}\sum_{t=1}^{T}\sum_{b\in\mathcal{B}}\hat{\lambda}_b\ln\big(1+\gamma_b(t)\big),$
and the expectation constraints in problem \eqref{Eq: Lyapunov network problem} will be satisfied as long as the corresponding virtual queues are mean rate stable \cite{Neely/Stochastic}. By incorporating the stability requirements for all virtual queues, problem \eqref{Eq: Lyapunov network problem} can be equivalently rewritten as
\begin{subequations}\label{Eq: Realized-based obj}
\begin{IEEEeqnarray}{cl}
\hspace{-2.5em}\underset{\boldsymbol{\alpha}(t),\tilde{\theta}_b( \mathbf{h}_b(t);t),\gamma_b(t)}{\text{maximize}}&\textstyle~~ \lim\limits_{T\to\infty}\frac{1}{T}\sum\limits_{t=1}^{T}\sum\limits_{b\in\mathcal{B}}\hat{\lambda}_b\ln\big(1+\gamma_b(t)\big)
\\\hspace{-2.5em}\mbox{subject to}&~~\mbox{Stability of \eqref{Eq: VQ-Y}, \eqref{Eq: VQ-Z}, \eqref{Eq: VQ-D}, \eqref{Eq: VQ-F}},\notag
\\\hspace{-2.5em}&~~\boldsymbol{\alpha}(t)\in\mathcal{A},~\forall\,t,
\\\hspace{-2.5em}&~~\tilde{\theta}_b(\mathbf{h}_b(t);t)\in[0, v_{b}^{\rm max}],~\forall\,t,b\in\mathcal{B},
\\\hspace{-2.5em}&~~\gamma_b(t)\in[0, v_{b}^{\rm max}],~\forall\,t,b\in\mathcal{B}.
\end{IEEEeqnarray}
\end{subequations}
Here, we note that although problems \eqref{Eq: Estimated auxiliary network problem} and \eqref{Eq: Realized-based obj}  are equivalent, the optimal solution of the former problem is found in one time instant whereas the optimum of the latter incurs a larger delay.

For notational simplicity,  $\mathbf{\Xi}(t)=\big(Y^{\boldsymbol{\chi}_b}_{\boldsymbol{\omega}_b}(t),Z_{b}(t),D_{b}(t),F_{b}(t),b\in\mathcal{B},\boldsymbol{\omega}_b\in\mathcal{W}_b,\boldsymbol{\chi}_b \in\mathcal{A}_b\big)$ denotes the combined queue vector. Then, we express the conditional Lyapunov drift-plus-penalty for time slot $t$ as
\begin{equation}
\textstyle\mathbb{E}\big[L(\mathbf{\Xi}(t+1))-L(\mathbf{\Xi}(t))-W\sum\limits_{b\in\mathcal{B}}\hat{\lambda}_b\ln\big(1+\gamma_b(t)\big)\big|\mathbf{\Xi}(t)\big],\label{Eq: Drift and penalty}
\end{equation}
where  
%
%
%
 $ L(\mathbf{\Xi}(t))=\frac{1}{2}\sum_{b\in\mathcal{B}}\big(\sum_{\boldsymbol{\omega}_b\in\mathcal{W}_b}\sum_{\boldsymbol{\chi}_b\in\mathcal{A}_b}Y^{\boldsymbol{\chi}_b}_{\boldsymbol{\omega}_b}(t)^2+Z_{b}(t)^2
+D_{b}(t)^2+F_{b}(t)^2\big)$
%
%
%
is the Lyapunov function, and $\kappa\geq 0$ is the tradeoff factor between utility optimality and queue stability. Subsequently, applying \eqref{Eq: VQ-Y}, \eqref{Eq: VQ-Z}, \eqref{Eq: VQ-D}, \eqref{Eq: VQ-F}, and $(\max\{x,0\})^2\leq x^2$ to \eqref{Eq: Drift and penalty}, we derive
\begin{figure*}
\begin{align*}
&\textstyle K_0[k] =\sum\limits_{b\in\mathcal{B}}\sum\limits_{m\in\mathcal{M}_b}\sum\limits_{s\in\mathcal{S}}\Bigg[ \sum\limits_{\boldsymbol{\chi}_b\in\mathcal{A}_b }Y^{\boldsymbol{\chi}_b}_{\mathbf{h}_b}\notag
\ln\bigg(1+\frac{\chi_{bm}^{(s)}h_{bm}^{(s)}}{\sigma^2} +\sum\limits_{b'\in\mathcal{B}\setminus b}\sum\limits_{ m'\in\mathcal{M}_{b'}}\frac{\hat{P}_{b'm'}^{(s)} [k] [\mathcal{H}_{b'm}^{(s)}]_{\rm max}}{\sigma^2}  \bigg)-F_b\cdot\mathbbm{1}\big\{F_{b}<0\big\}\notag
\\&\textstyle 
\times\ln\bigg(1+\frac{\hat{P}_{bm}^{(s)}[k]h_{bm}^{(s)}}{\sigma^2}+\sum\limits_{b'\in\mathcal{B}\setminus b\atop m'\in\mathcal{M}_{b'}}  \frac{\hat{P}_{b'm'}^{(s)}[k][\mathcal{H}_{b'm}^{(s)}]_{\rm max}}{\sigma^2}  \bigg)
+\big(D_{b}+F_{b}\cdot\mathbbm{1}\big\{F_{b}>0\big\}+Z_{b}\big)
\ln\bigg(1+\sum\limits_{b'\in\mathcal{B}\setminus b\atop m'\in\mathcal{M}_{b'}} \frac{\hat{P}_{b'm'}^{(s)}[k][\mathcal{H}_{b'm}^{(s)}]_{\rm max} }{\sigma^2}\bigg)\bigg],
\\&K_{bm}^{(s)}[k]\textstyle=\sum\limits_{b'\in\mathcal{B}\setminus b}\sum\limits_{ m'\in\mathcal{M}_{b'}}\Bigg[\frac{\big(D_{b'}+F_{b'}\cdot\mathbbm{1}\{F_{b'}>0\}+Z_{b'}\big)[\mathcal{H}_{bm'}^{(s)}]_{\rm max}}{\Big(\sigma^2+\sum\limits_{i\in\mathcal{B}\setminus b'}\sum\limits_{ j\in\mathcal{M}_{i}} \hat{P}_{ij}^{(s)}[k][\mathcal{H}_{im'}^{(s)}]_{\rm max}  \Big)}
 +\sum\limits_{\boldsymbol{\chi}_{b'}\in\mathcal{A}_{b'}}\frac{Y^{\boldsymbol{\chi}_{b'}}_{\mathbf{h}_{b'}}[\mathcal{H}_{bm'}^{(s)}]_{\rm max}}{\Big(\sigma^2 +\chi_{b'm'}^{(s)}h_{b'm'}^{(s)} +\sum\limits_{i\in\mathcal{B}\setminus b'}\sum\limits_{ j\in\mathcal{M}_{i}}\hat{P}_{ij}^{(s)} [k][\mathcal{H}_{im'}^{(s)}]_{\rm max}\Big)}\notag
\\&\textstyle \hspace{4.5em}-\frac{[\mathcal{H}_{bm'}^{(s)}]_{\rm max}F_{b'}\cdot\mathbbm{1}\{F_{b'}<0\}}{\Big(\sigma^2+\hat{P}_{b'm'}^{(s)}[k]h_{b'm'}^{(s)}+\sum\limits_{i\in\mathcal{B}\setminus b'}\sum\limits_{ j\in\mathcal{M}_{i}} \hat{P}_{ij}^{(s)}[k][\mathcal{H}_{im'}^{(s)}]_{\rm max}  \Big)}\Bigg]
-\frac{h_{bm}^{(s)}F_{b'}\cdot\mathbbm{1}\{F_{b'}<0\}}{\Big(\sigma^2+\hat{P}_{bm}^{(s)}[k]h_{bm}^{(s)} +\sum\limits_{b'\in\mathcal{B}\setminus b}\sum\limits_{ m'\in\mathcal{M}_{b'}}\hat{P}_{b'm'}^{(s)}[k] [\mathcal{H}_{b'm}^{(s)}]_{\rm max} \Big)}.
\end{align*}
\noindent\makebox[\linewidth]{\rule{0.84\paperwidth}{0.4pt}}
\end{figure*}
\begin{align}
& \eqref{Eq: Drift and penalty}\leq
\textstyle \frac{1}{2}\sum\limits_{b\in\mathcal{B}}|\mathcal{M}_b|^2\lambda_{\rm max}^2+\sum\limits_{b\in\mathcal{B}}\big(\frac{1}{2}|\mathcal{A}_b|+\frac{3}{2}\big)\big(|\mathcal{M}_b||\mathcal{S}|R_{\rm max}\big)^2\notag
 \\&+\sum\limits_{b\in\mathcal{B}}\mathbb{E}\bigg[\sum\limits_{\boldsymbol{\chi}_b\in\mathcal{A}_b}Y^{\boldsymbol{\chi}_b}_{\mathbf{h}_b(t)}(t)\big(\tilde{v}_{b}(\mathbf{h}_b(t),\boldsymbol{\chi}_b;t)-\tilde{\theta}_b(\mathbf{h}_b(t);t)\big)\notag
\\& 
+Z_b(t)\big(\tilde{\theta}_b(\mathbf{h}_b(t);t)-v_{b}(t)\big)  +D_{b}(t)\big(\hat{\lambda}_b-v_{b}(t)\big)\notag
\\& +F_{b}(t)\big(\gamma_b(t)-v_{b}(t)\big)
-\kappa\hat{\lambda}_b\ln\big(1+\gamma_b(t)\big)\big|\mathbf{\Xi}(t)\Big].\label{Eq: Lyapunov bound-1}
\end{align}
The solution to \eqref{Eq: Realized-based obj} can be obtained by minimizing the upper bound of the conditional Lyapunov drift-plus-penalty, i.e.,  \eqref{Eq: Lyapunov bound-1}, in each time slot $t$ \cite{Neely/Stochastic}. In this regard, we have three decomposed sub-problems which are detailed as follows.

The first decomposed problem is
\begin{equation*}
\mbox{{\bf P1:}}~~ \underset{ 0\leq \gamma_b(t)\leq v_{b}^{\rm max}}{\text{minimize}}~~F_{b}(t)\gamma_b(t)-\kappa\hat{\lambda}_b\ln\big(1+\gamma_b(t)\big)
\end{equation*}
 $\forall\,b\in\mathcal{B},$ with the solution
\begin{equation*}
\begin{aligned}
\gamma_b(t)=
\begin{cases}
v_{b}^{\rm max},&  \mbox{if~}F_{b}(t)\leq \frac{\kappa\hat{\lambda}_b}{v_{b}^{\rm max}+1},
\\\frac{\kappa\hat{\lambda}_b}{F_{b}(t)}-1,&  \mbox{if~}\frac{\kappa\hat{\lambda}_b}{v_{b}^{\rm max}+1}< F_{b}(t)\leq \kappa\hat{\lambda}_b,
\\0,&  \mbox{otherwise.}
\end{cases}
\end{aligned}
\end{equation*}
The second decomposed problem is
\begin{IEEEeqnarray*}{ccl}
\mbox{{\bf P2:}}~~&\underset{\tilde{\theta}_b( \mathbf{h}_b(t);t)}{\text{minimize}} & \textstyle~~\Big(Z_{b}(t)-\sum\limits_{\boldsymbol{\chi}_b\in\mathcal{A}_b}Y^{\boldsymbol{\chi}_b}_{\mathbf{h}_b(t)}(t)\Big)\tilde{\theta}_b( \mathbf{h}_b(t);t)\label{Eq: Second subproblem}
\\&\mbox{subject to}&~~\eqref{Eq: Auxiliary function-1},
\end{IEEEeqnarray*}
for all $b\in\mathcal{B},$ where the solution is
\begin{equation*}
\begin{aligned}
\tilde{\theta}_b(\mathbf{h}_b(t);t)=
\begin{cases}
v_{b}^{\rm max},&  \mbox{if~}
Z_{b}(t)<\sum\limits_{\boldsymbol{\chi}_b \in\mathcal{A}_b}Y^{\boldsymbol{\chi}_b}_{\mathbf{h}_b(t)}(t),
\\0,&  \mbox{otherwise.}
\end{cases}
\end{aligned}
\end{equation*}
Finally, we have the third decomposed problem as
\begin{subequations}\label{Eq: Third subproblem objective}
\begin{align}
&\textstyle\mbox{{\bf P3:}}~~\underset{\boldsymbol{\alpha}(t)\in\mathcal{A}}{ \text{minimize}} ~\sum\limits_{b\in\mathcal{B}}\sum\limits_{m\in\mathcal{M}_b}\sum\limits_{s\in\mathcal{S}}\Bigg[
\sum\limits_{\boldsymbol{\chi}_b\in\mathcal{A}_b }Y^{\boldsymbol{\chi}_b}_{\mathbf{h}_b(t)}(t)\notag
\\&\textstyle\times\ln\bigg(1+\frac{\chi_{bm}^{(s)}h_{bm}^{(s)}(t)}{\sigma^2} +\sum\limits_{b'\in\mathcal{B}\setminus b\atop m'\in\mathcal{M}_{b'}} \frac{ P_{b'm'}^{(s)}(t)[\mathcal{H}_{b'm}^{(s)}]_{\rm max}}{\sigma^2}  \bigg)\label{Eq: Third subproblem objective-1}
\\&\textstyle -F_{b}(t)\cdot\mathbbm{1}\big\{F_{b}(t)<0\big\}\notag
 \\&\textstyle\times\ln\bigg(1+\frac{P_{bm}^{(s)}(t)h_{bm}^{(s)}(t)}{\sigma^2}+\sum\limits_{b'\in\mathcal{B}\setminus b\atop m'\in\mathcal{M}_{b'}} \frac{P_{b'm'}^{(s)}  (t)[\mathcal{H}_{b'm}^{(s)}]_{\rm max}}{\sigma^2} \bigg)\label{Eq: Third subproblem objective-2}
\\&\textstyle +\big(D_{b}(t)+F_{b}(t)\cdot\mathbbm{1}\big\{F_{b}(t)>0\big\}+Z_{b}(t)\big)\notag
 \\&\textstyle\times
\ln\bigg(1+\sum\limits_{b'\in\mathcal{B}\setminus b}\sum\limits_{ m'\in\mathcal{M}_{b'}} \frac{P_{b'm'}^{(s)}(t)[\mathcal{H}_{b'm}^{(s)}]_{\rm max} }{\sigma^2}\bigg)\label{Eq: Third subproblem objective-3}
\\&\textstyle -\big(D_{b}(t)+F_{b}(t)\cdot\mathbbm{1}\big\{F_{b}(t)>0\big\}+Z_{b}(t)\big)\notag
\\&\textstyle\times\ln\bigg(1+\frac{P_{bm}^{(s)}(t)h_{bm}^{(s)}(t)}{\sigma^2}+\sum\limits_{b'\in\mathcal{B}\setminus b\atop m'\in\mathcal{M}_{b'}} \frac{P_{b'm'}^{(s)}(t)[\mathcal{H}_{b'm}^{(s)}]_{\rm max}}{\sigma^2}  \bigg)\label{Eq: Third subproblem objective-4}
\\&\textstyle -\Big(\sum\limits_{\boldsymbol{\chi}_b\in\mathcal{A}_b }Y^{\boldsymbol{\chi}_b}_{\mathbf{h}_b(t)}(t)-F_{b}(t)\cdot\mathbbm{1}\big\{F_{b}(t)<0\big\}\Big)\notag
 \\&\textstyle\times\ln\bigg(1+\sum\limits_{b'\in\mathcal{B}\setminus b}\sum\limits_{ m'\in\mathcal{M}_{b'}}\frac{P_{b'm'}^{(s)} (t) [\mathcal{H}_{b'm}^{(s)}]_{\rm max}}{\sigma^2}  \bigg)\bigg]
\label{Eq: Third subproblem objective-5}
\end{align}
\end{subequations}
which is an  NP-hard  integer programming problem due to the indicator function and the finite set $\mathcal{L}_b$ in $\mathcal{A}.$ In order to have a tractable analysis, we relax the power constraints in the global action space  $\mathcal{A}$ as
\begin{align}\label{Eq: Relaxed power}
\boldsymbol{\alpha}(t)\in\mathcal{A}\Longrightarrow
\begin{cases}
P_{bm}^{(s)}(t)\geq 0,~\forall\,b\in\mathcal{B},m\in\mathcal{M}_b,s\in\mathcal{S},
\\\sum\limits_{m\in\mathcal{M}_b}\sum\limits_{s\in\mathcal{S}} P_{bm}^{(s)}(t)\leq |\mathcal{S}|\totalpower b,~\forall\,b\in\mathcal{B}.
\end{cases}
\end{align}
However, since \eqref{Eq: Third subproblem objective-1}-\eqref{Eq: Third subproblem objective-3} are concave functions whereas \eqref{Eq: Third subproblem objective-4} and \eqref{Eq: Third subproblem objective-5} are convex functions, {\bf P3}, with the relaxed power constraints \eqref{Eq: Relaxed power}, belongs to the family of difference of convex programming problems which cannot be solved by  standard approaches in convex optimization. In this regard, the convex-concave procedure (CCP) provides an iterative solution  which converges to a locally optimal solution \cite{Lipp:2015:CCP}.

\subsubsection{Convex-concave Procedure for the Decomposed Problem}
Let us  detail the derivation of the CCP as follows. We first denote the sum of concave functions as 
\begin{equation}\label{Eq: Concave function}
f(\boldsymbol{\alpha})=\eqref{Eq: Third subproblem objective-1}+\eqref{Eq: Third subproblem objective-2}+\eqref{Eq: Third subproblem objective-3},
\end{equation}
which will be convexified by the first order Taylor approximation.
In iteration $k,$ we first select a feasible global action $\hat{\boldsymbol{\alpha}}[k]$ from the relaxed power constraints \eqref{Eq: Relaxed power} and convexify the concave function $f(\boldsymbol{\alpha})$ according to
\begin{align}\label{Eq: Convexified function}
&\hat{f}(\boldsymbol{\alpha};\hat{\boldsymbol{\alpha}}[k])  \coloneqq f(\hat{\boldsymbol{\alpha}}[k])+\nabla f(\hat{\boldsymbol{\alpha}}[k])(\boldsymbol{\alpha}-\hat{\boldsymbol{\alpha}}[k])^{\rm T}\notag
\\&\quad \textstyle=K_0[k]+\sum\limits_{b\in\mathcal{B}}\sum\limits_{m\in\mathcal{M}_b}\sum\limits_{s\in\mathcal{S}}K_{bm}^{(s)}[k]\Big(P_{bm}^{(s)}-\hat{P}_{bm}^{(s)}[k]\Big),
\end{align}
where the superscript ${\rm T}$ denotes the transpose of the vector. $K_0[k]$ and $K_{bm}^{(s)}[k]$ are constants whose derivations are given on the top of this page.
Then, substituting \eqref{Eq: Convexified function} for \eqref{Eq: Concave function}, we have the convexified optimization problem
\begin{subequations}\label{Eq: Convexify decomposed}
\begin{IEEEeqnarray}{cl}
%
 \underset{\boldsymbol{\alpha}}{ \text{minimize}} &\textstyle~ \hat{f}(\boldsymbol{\alpha};\hat{\boldsymbol{\alpha}}[k])-\sum\limits_{b\in\mathcal{B}}\sum\limits_{m\in\mathcal{M}_b}\sum\limits_{s\in\mathcal{S}}\notag
\\& \textstyle~\Big[\big(D_{b}+F_{b}\cdot\mathbbm{1}\big\{F_{b}>0\big\}+Z_{b}\big)\notag
\\& \textstyle~\times\ln\Big(1+\frac{P_{bm}^{(s)}h_{bm}^{(s)}}{\sigma^2}+\sum\limits_{b'\in\mathcal{B}\setminus b\atop m'\in\mathcal{M}_{b'}} \frac{P_{b'm'}^{(s)}[\mathcal{H}_{b'm}^{(s)}]_{\rm max}}{\sigma^2}  \Big)\notag
\\&\textstyle  ~+\Big(\sum\limits_{\boldsymbol{\chi}_b\in\mathcal{A}_b }Y^{\boldsymbol{\chi}_b}_{\mathbf{h}_b}-F_{b}\cdot\mathbbm{1}\big\{F_{b}<0\big\}\Big)\notag
\\&\textstyle  \times \ln\bigg(1+\sum\limits_{b'\in\mathcal{B}\setminus b}\sum\limits_{ m'\in\mathcal{M}_{b'}}\frac{P_{b'm'}^{(s)}  [\mathcal{H}_{b'm}^{(s)}]_{\rm max}}{\sigma^2}  \bigg)\bigg]
\\\mbox{subject to}&~P_{bm}^{(s)}\geq 0,~\forall\,b\in\mathcal{B},m\in\mathcal{M}_b,s\in\mathcal{S},
\\&\textstyle ~\sum\limits_{m\in\mathcal{M}_b}\sum\limits_{s\in\mathcal{S}} P_{bm}^{(s)}\leq |\mathcal{S}|\totalpower b,~\forall\,b\in\mathcal{B}.
\end{IEEEeqnarray}
\end{subequations}
Due to the unavailability of a tractable solution to problem \eqref{Eq: Convexify decomposed}, we resort to CVX to numerically find the optimal power which is denoted by $\boldsymbol{\alpha}^{*}_{\rm ccp}[k].$ Subsequently, we set $\boldsymbol{\alpha}^{*}_{\rm ccp}[k]$ as the reference point for convexifying $f(\boldsymbol{\alpha})$ in the next iteration, i.e., iteration $k+1.$ The iterative procedure stops when the predetermined stopping criterion is satisfied. In the above explanation of the CCP, we remove the time-dependent index $t$ in order to avoid abuse of notation, but the CCP algorithm for solving ${\bf P3}$ is executed for each $t$. The steps of the CCP algorithm are shown in Algorithm \ref{Alg: CCP}. After obtaining the converged solution from the CCP algorithm, the controller finds $\boldsymbol{\alpha}(t)\in\mathcal{A},$ with respect to the nearest Euclidean distance $||\boldsymbol{\alpha}(t)-\boldsymbol{\alpha}_{\rm ccp}^{*}||$,  which is the control recommendation based on the current state realization $\mathbf{h}(t).$ 
Note that the decoupled problems{ \bf P1}, {\bf P2}, and {\bf P3}, and all virtual queues are solved and updated in each time slot $t.$

\begin{algorithm}[t]
  \caption{CCP for solving {\bf P3}.}
  \begin{algorithmic}[1]
    \State Initialize a feasible point $\hat{\boldsymbol{\alpha}}[0]$ from \eqref{Eq: Relaxed power} and  $k=0.$
    \Repeat
      \State Convexify $f(\boldsymbol{\alpha})$ by $\hat{f}(\boldsymbol{\alpha};\hat{\boldsymbol{\alpha}}[k]).$
      \State Solve problem \eqref{Eq: Convexify decomposed} and denote the optimal solution as $\boldsymbol{\alpha}^{*}_{\rm ccp}[k].$
      \State Update $\hat{\boldsymbol{\alpha}}[k+1]=\boldsymbol{\alpha}^{*}_{\rm ccp}[k]$ and $k\leftarrow k+1.$
    \Until{Stopping criteria are satisfied.}
  \end{algorithmic}\label{Alg: CCP}
\end{algorithm}

However, since the fronthaul is activated every $T_0$ slots, a real-time state realization is not available at the controller for solving {\bf P3}, and none of the virtual queues required in {\bf P3} can be updated in real-time. To tackle this issue, we first note that as time evolves, the time-averaged queue value will converge, and the instantaneous queue length will approach the converged value \cite{Online_learning}. Hence, instead of the instantaneous queue value, the controller considers the time-averaged queue values up to slot $t=(n-1)T_0$ (i.e., the end slot of the last time frame $n-1$). That is, we substitute
\begin{equation}\label{Eq: Average queue for CCP}
\begin{cases}
\frac{1}{(n-1)T_0}\sum\limits_{\tau=1}^{(n-1)T_0}Y^{\boldsymbol{\chi}_b}_{\boldsymbol{\omega}_b}(\tau),
\\\frac{1}{(n-1)T_0}\sum\limits_{\tau=1}^{(n-1)T_0}Z_{b}(\tau),
\end{cases}
\begin{cases}
\frac{1}{(n-1)T_0}\sum\limits_{\tau=1}^{(n-1)T_0}D_{b}(\tau),
\\\frac{1}{(n-1)T_0}\sum\limits_{\tau=1}^{(n-1)T_0}F_{b}(\tau),
\end{cases}
\end{equation}
into $Y^{\boldsymbol{\chi}_b}_{\boldsymbol{\omega}_b}(t)$, $Z_b(t)$, $D_b(t)$, and $F_b(t)$ in {\bf P3}, respectively. Subsequently, the controller finds a global action realization $\boldsymbol{\alpha}_{\rm II}\in\mathcal{A}$ for each global state realizations $\boldsymbol{\omega}\in\mathcal{W}$ using Algorithm \ref{Alg: CCP} and \eqref{Eq: Average queue for CCP} for the queue values.
Similar to Section \ref{Sec: Empirical information}, from all global states and the corresponding CCP-converged transmit power $\boldsymbol{\alpha}_{\rm II},$ the controller finds the corresponding mapping rule $\Psi^{*}_{\rm II}(n):\mathcal{W}\to\mathcal{A},$ i.e., $\Psi^{*}_{\rm II}(\boldsymbol{\omega};n)=\boldsymbol{\alpha}_{\rm II},$ and feeds this rule back to BSs. In the remaining time slots of frame $n,$ the BSs serve the UEs and gather the state realizations $\{\mathbf{h}_b(t),t\in\mathcal{T}(n)\}.$ 
At the beginning of  the next frame $n+1,$ each BS $b$ uploads the $T_0$ gathered state realizations during time frame $a.$ Having the past state realization information, the controller solves { \bf P1}, {\bf P2}, and {\bf P3}, and updates  \eqref{Eq: VQ-Y}, \eqref{Eq: VQ-Z}, \eqref{Eq: VQ-D}, and \eqref{Eq: VQ-F}, for time slots $t\in\mathcal{T}(n)$. Finally, the controller calculates \eqref{Eq: Average queue for CCP} up to $t=nT_0$ and finds the recommendations for time frame $n+1$.
The steps of the controller's resource allocation procedures at the beginning of time frame $n$ are provided in Algorithm \ref{Alg: Approach II}.
Analogous to Section \ref{Sec: Empirical information}, we assume that BSs' state spaces and all possible state-to-action mapping rules are known at both the controller and all BSs, and each possible state realization and mapping rule are represented by the real values. Given the required transmission rate $R_{\rm unit}$ for sending a single real value, we have $(T_0+1)R_{\rm unit}$ for $R^{\rm U}$ in \eqref{Eq: FH time cost}, where the value 1 accounts for the estimated mean data arrival $\hat{\lambda}_b,$  and $R^{\rm F}=R_{\rm unit}$ in \eqref{Eq: FH download cost}.
For the computational complexity in the state realization-based approach, {\bf P1} and {\bf P2} have $2|\mathcal{B}|$ sub-problems in which each sub-problem has one optimization variable and two affine constraints.
Additionally, there are $ |\mathcal{S}|\big( \sum_{b\in\mathcal{B}}|\mathcal{M}_b|\big)$ optimization variables and $|\mathcal{B}|+ |\mathcal{S}| \big(\sum_{b\in\mathcal{B}}|\mathcal{M}_b|\big)$ affine constraints in the convexified problem (36) of {\bf P3}. The computational complexity increases linearly with the number of BSs, UEs, and sub-carriers. In contrast with the statistics-based approach, the state realization-based approach has a lighter burden in the fronthaul and a lighter computational complexity at the SDN controller.
The information exchange in the fronthaul of the realization-based resource allocation approach is summarized as follows. In the beginning time portion $\varsigma(n)$ of each frame $n$, the BSs upload the state realizations during last frame $n-1$, i.e., $\{\mathbf{h}_b(t),t\in\mathcal{T}(n-1)\},\forall\,b\in\mathcal{B},$ and the empirically estimated mean data arrivals to the controller. Then, the controller feeds back $\Psi^{*}_{\rm II}(n)$ to the BSs.

The above two resource allocation approaches achieve the same performance given the identical statistics of channels and the time penalty. However, due to different burdens on the fronthaul which affect the time penalty, the two approaches do not necessarily provide the identical optimal utility.

\begin{algorithm}[t]
  \caption{State realization-based resource allocation.}
  \begin{algorithmic}[1]
    \Require
      $\{\hat{\lambda}_b,\mathbf{h}_b(t),b\in\mathcal{B},t\in\mathcal{T}(n-1)\},$ and all virtual queues  (i.e., \eqref{Eq: VQ-Y}, \eqref{Eq: VQ-Z}, \eqref{Eq: VQ-D}, and \eqref{Eq: VQ-F}) for $t=(n-2)T_0+1.$
    \Ensure
      Optimal $\Psi^{*}_{\rm II}(n).$
    \State Initialize $\tau=(n-2)T_0+1.$
    \Repeat
      \State Solve {\bf P1}, {\bf P2}, and {\bf P3} for $t=\tau.$
      \State Update \eqref{Eq: VQ-Y}, \eqref{Eq: VQ-Z}, \eqref{Eq: VQ-D}, and \eqref{Eq: VQ-F} for $t=\tau+1$ and $\tau\leftarrow\tau+1.$
    \Until $\tau>(n-1)T_0.$
\State Solve {\bf P3} for each $\boldsymbol{\omega}\in\mathcal{W}$ with \eqref{Eq: Average queue for CCP} and  find the mapping rule  $\Psi^{*}_{\rm II}(n):\mathcal{W}\to\mathcal{A}.$
  \end{algorithmic}\label{Alg: Approach II}
\end{algorithm}

\subsection{User Scheduling at the Base Station}\label{Sec: Scheduling}

After observing $\boldsymbol{\omega}_b=\mathbf{h}_b(t)$ in each slot $t$ and referring to the received mapping rule $\Psi^{*}_{\rm I}(t)$ (or $\Psi^{*}_{\rm II}(n)$), BS $b$ obtains a suggested transmit power vector $\boldsymbol{\alpha}^{*}_b(t)\coloneqq (\alpha_{bm}^{(s)*},m\in\mathcal{M}_b,s\in\mathcal{S})\in\mathcal{A}_b$ to serve the UEs  \cite{Neely:13:Arxiv}. 
Although the suggested action $\boldsymbol{\alpha}^{*}_{b}(t)$ achieves a DL rate which is higher than the data arrival rate, the controller might suggest that the BS serves the UEs with better channel quality while deferring the service for UEs with higher queue lengths incurring high latency. To address this  concern, BS $b$ schedules its UEs based on the queue lengths and optimizes the transmit power over the set of allocated sub-carriers 
 $\mathcal{X}_b(t)=\{s\in\mathcal{S}|\sum_{m\in\mathcal{M}_b}\alpha_{ bm}^{(s)*}(t)> 0\}$.
Accordingly, incorporating the available sub-carriers, we remodel the BS's objective as
\begin{subequations}\label{Eq: aided-bs-obj}
\begin{IEEEeqnarray}{cl}
\hspace{-2em}\underset{P^{(s)}_{bm}(t)\in\mathcal{L}_b}{\mbox{maximize}}&\textstyle ~~\sum\limits_{m\in\mathcal{M}}\sum\limits_{s\in\mathcal{S}}\bar{R}_{bm}^{(s)}
\\\hspace{-2em}\mbox{subject to}& ~~\lim\limits_{t\to\infty}\frac{\mathbb{E}\left[|Q_{bm}(t)|\right]}{t}\to 0,~\forall\,m\in\mathcal{M}_b,
\\\hspace{-2em}&\textstyle~~\sum\limits_{m\in\mathcal{M}_b}\mathbbm{1}\big\{ P_{bm}^{(s)}(t)>0\big\}\leq 1,~\forall\,t,s\in\mathcal{X}_b(t),
\\\hspace{-2em}&\textstyle~~\sum\limits_{m\in\mathcal{M}_b}\sum\limits_{s\in\mathcal{S}} P_{bm}^{(s)}(t)\leq |\mathcal{S}|\totalpower b,~\forall\,t,
\\\hspace{-2em}&~~P_{bm}^{(s)}(t)=0,~\forall\,t,m\in\mathcal{M}_b,s\notin\mathcal{X}_b(t),
\end{IEEEeqnarray}
\end{subequations}
which can be solved by dynamic programming. However, dynamic programing suffers from the curse of dimensionality. To address this, we resort to tools of Lyapunov optimization.
Given the combined queue vector $\mathbf{Q}_b(t)=(Q_{bm}(t), m\in\mathcal{M}_b),$ an upper bound on the conditional Lyapunov drift-plus-penalty for slot $t$ is given by
\begin{multline*}
 \textstyle\mathbb{E}\Big[L(\mathbf{Q}_b(t+1))-L(\mathbf{Q}_b(t))-\sum\limits_{m\in\mathcal{M}_b}\sum\limits_{ s\in\mathcal{S}}VR_{bm}^{(s)}(t)
\Big|\mathbf{Q}_b(t)\Big]
\\\textstyle\leq\frac{|\mathcal{M}_b|\big(|\mathcal{S}|^2R^2_{{\rm max}}+A_{{\rm max}}^2\big)}{2}-\mathbb{E}\bigg[\sum\limits_{m\in\mathcal{M}_b }\sum\limits_{ s\in\mathcal{S}}VR_{bm}^{(s)}(t)
\\  \textstyle
+\sum\limits_{m\in\mathcal{M}_b}Q_{bm}(t) \Big(\sum\limits_{s\in\mathcal{S}}R_{bm}^{(s)}(t)-\lambda_{bm}(t)\Big)
\big|\mathbf{Q}_b(t)\Big],
\end{multline*}
where the parameter $V$ reveals the tradeoff between queue stability and sum rate maximization. Analogously to \eqref{Eq: Realized-based obj},  BS $b$ minimizes the upper bound on the conditional Lyapunov drift-plus-penalty by solving
\begin{subequations}\label{Eq: BS NP-hard problem}
\begin{IEEEeqnarray}{cl}
\hspace{-2em} \underset{ P_{bm}^{(s)}\in\mathcal{L}_b}{\mbox{maximize}}& \textstyle~~\sum\limits_{m\in\mathcal{M}_b}\sum\limits_{s\in\mathcal{S}}(Q_{bm}+V)
\cdot\mathbb{E}_{I_{bm}^{(s)}}  \big[R_{bm}^{(s)}\big|\boldsymbol{\omega}_b,\mathcal{X}_b\big]\label{Eq: BS NP-hard problem-1}
\\\hspace{-2em}\mbox{subject to}&\textstyle~~\sum\limits_{m\in\mathcal{M}_b}\mathbbm{1}\big\{ P_{bm}^{(s)}>0\big\}\leq 1,~\forall\,t,s\in\mathcal{X}_b,\label{Eq: BS NP-hard problem-3}
\\\hspace{-2em}&\textstyle~~\sum\limits_{m\in\mathcal{M}_b}\sum\limits_{s\in\mathcal{S}} P_{bm}^{(s)}\leq |\mathcal{S}|\totalpower b,\label{Eq: BS NP-hard problem-4}
\\\hspace{-2em}& ~~P_{bm}^{(s)}=0,~\forall\,m\in\mathcal{M}_b,s\notin\mathcal{X}_b.\label{Eq: BS NP-hard problem-2}
\end{IEEEeqnarray}
\end{subequations}
 Note that \eqref{Eq: BS NP-hard problem}  is solved in each time slot $t$ although the time index is omitted for the sake of simplicity. Here, the expectation is over the conditional aggregate interference $I_{bm}^{(s)}=\sum_{b'\in\mathcal{B}\backslash b}\sum_{m'\in\mathcal{M}_{b'}}P_{b'm'}^{(s)}  h_{b'm}^{(s)}$ with the estimated distribution $\hat{\Pr}\big(I_{bm}^{(s)};t\big|\boldsymbol{\omega}_b,\mathcal{X}_b)$ in slot $t.$
However, due to the indicator function in \eqref{Eq: BS NP-hard problem-3} and the finite set $\mathcal{L}_b$, \eqref{Eq: BS NP-hard problem} is an  NP-hard  integer programming problem. Analogously to {\bf P3}, we rewrite \eqref{Eq: BS NP-hard problem}  as
\begin{subequations}\label{Eq: aided-bs-Lya-obj}
\begin{IEEEeqnarray}{cl}
 \underset{ P_{bm}^{(s)}}{\mbox{maximize}}&\textstyle ~~\sum\limits_{m\in\mathcal{M}_b}\sum\limits_{s\in\mathcal{S}}(Q_{bm}+V)
\\&\textstyle ~~\quad\times\mathbb{E}_{I_{bm}^{(s)}}  \bigg[\ln\bigg(1+\frac{P_{bm}^{(s)}h_{bm}^{(s)}}{\sigma^2+ I_{bm}^{(s)}} \bigg)\Big|\boldsymbol{\omega}_b,\mathcal{X}_b\bigg]\label{Eq: aided-bs-Lya-obj-1}
\\\mbox{subject to}& \textstyle~~\sum\limits_{m\in\mathcal{M}_b}\sum\limits_{s\in\mathcal{S}}P_{bm}^{(s)} \leq |\mathcal{S}|\totalpower{b},\label{Eq: aided-bs-Lya-obj-2}
\\&~~ P_{bm}^{(s)}  \geq 0,~\forall\,m\in\mathcal{M}_b,s\in\mathcal{X}_b,\label{Eq: aided-bs-Lya-obj-3}
\\&~~ P_{bm}^{(s)}  = 0,~\forall\,m\in\mathcal{M}_b,s\notin\mathcal{X}_b,\label{Eq: aided-bs-Lya-obj-4}
\end{IEEEeqnarray}
\end{subequations}
with the relaxed linear power constraint.
\begin{lemma}\label{lemma power}
The optimal solution  to \eqref{Eq: aided-bs-Lya-obj} is detailed as follows.  For all $m\in\mathcal{M}_b$ and $s\in\mathcal{X}_b$, the optimal $P_{bm}^{(s)*}>$   is found by
\begin{align}
\textstyle\mathbb{E}_{I_{bm}^{(s)}}\bigg[\frac{\big(Q_{bm}+V\big)h_{bm}^{(s)}}{\sigma^2+ I_{bm}^{(s)}+P_{bm}^{(s)*}h_{bm}^{(s)}}
 \bigg|\boldsymbol{\omega}_b,\mathcal{X}_b\bigg] =\gamma_b\label{Eq: BS water-willing}
\end{align}
if $\mathbb{E}_{I_{bm}^{(s)}}\Big[\frac{(Q_{bm}+V)h_{bm}^{(s)}}{\sigma^2+ I_{bm}^{(s)}}\Big|\boldsymbol{\omega}_b,\mathcal{X}_b \Big]>\gamma_b$.
Otherwise, $P_{bm}^{(s)*}=0.$
Here, the Lagrange multiplier $\gamma_b$ is chosen such that $\sum_{m\in\mathcal{M}_b}\sum_{s\in\mathcal{X}_b}P_{bm}^{(s)*} =|\mathcal{S}|\totalpower{b}.$
\end{lemma}
\begin{proof}
Please refer to Appendix \ref{Lem: opt pow}.
\end{proof}
Finally, denoting $\mathbf{P}_b^{*}=(P_{bm}^{(s)*},m\in\mathcal{M}_b,s\in\mathcal{S})$ and taking into account the set of available sub-carriers, we find the action $\mathbf{P}_b$ in $\mathcal{A}_b$ (with respect to the nearest Euclidean distance $||\mathbf{P}_b-\mathbf{P}_b^{*}||$) as the transmit power.
\begin{remark}\label{Remark: V impact}
From \eqref{Eq: BS water-willing}, we notice that for small $V,$ BS $b$ has high priority of serving  UE $m$ with a large queue length $Q_{bm}.$ 
Analogously, for large $V,$ the BS allocates higher power to the UEs with better links.
In the former case, low latency is achieved with the low DL rate. In the latter case, the BS maximizes throughput while allowing the queues of the UEs with worse links to grow.
\end{remark}
After receiving data from the BS, UEs feedback the aggregate interference to the BS. Then, the BS updates $Q_{bm}(t+1)$ according to \eqref{Eq: Queue-Q} and the estimated interference statistics for the next time slot $t+1$ as per
\begin{multline}
 \textstyle\hat{\Pr}\big(\tilde{I}_{bm}^{(s)};t+1|\boldsymbol{\omega}_b,\tilde{\mathcal{X}}_b\big)
= \frac{\mathbbm{1}\big\{\tilde{I}_{bm}^{(s)}=I_{bm}^{(s)}(t)\big\} }{1+\sum\limits_{\xi=1}^{t}\mathbbm{1}\{\boldsymbol{\omega}_b=\mathbf{h}_b(\xi),\tilde{\mathcal{X}}_b=\mathcal{X}_b(\xi)\}}\notag
\\\textstyle\quad+ \frac{\sum\limits_{\xi=1}^{t}\mathbbm{1}\{\boldsymbol{\omega}_b=\mathbf{h}_b(\xi),\tilde{\mathcal{X}}_b=\mathcal{X}_b(\xi)\}}{1+\sum\limits_{\xi=1}^{t}\mathbbm{1}\{\boldsymbol{\omega}_b=\mathbf{h}_b(\xi),\tilde{\mathcal{X}}_b=\mathcal{X}_b(\xi)\}}\cdot\hat{\Pr}\big(\tilde{I}_{bm}^{(s)};t|\boldsymbol{\omega}_b,\tilde{\mathcal{X}}_b\big),~\forall\,\tilde{I}_{bm}^{(s)},\notag
%
\end{multline}
if $\boldsymbol{\omega}_b=\mathbf{h}_b(\xi)$ and $\tilde{\mathcal{X}}_b=\mathcal{X}_b(\xi).$
 Otherwise,
\begin{align}\notag
 \hat{\Pr}\big(\tilde{I}_{bm}^{(s)};t+1|\boldsymbol{\omega}_b,\tilde{\mathcal{X}}_b\big)
= \hat{\Pr}\big(\tilde{I}_{bm}^{(s)};t|\boldsymbol{\omega}_b,\tilde{\mathcal{X}}_b\big),~\forall\,\tilde{I}_{bm}^{(s)}.
\end{align}
Moreover, the estimated mean arrival and state distribution for the next frame $n+1$ are empirically updated as follows:
\begin{align}
  \hat{\lambda}_b(n+1)&\textstyle=\frac{n\hat{\lambda}_b(n)}{n+1} + \sum\limits_{t\in\mathcal{T}(n)}\sum\limits_{m\in\mathcal{M}_b}\frac{\lambda_{bm}(t)}{(n+1)T_0},\label{Eq: Empirical data mean}
\\ \hat{\Pr}(\tilde{h}_{bm}^{(s)};n+1)  &\textstyle= \frac{a\hat{\Pr}(\tilde{h}_{bm}^{(s)};n) }{n+1}+  \sum\limits_{t\in\mathcal{T}(n)}\frac{\mathbbm{1}\{\tilde{h}_{bm}^{(s)}=h_{bm}^{(s)}(t)\}}{(n+1)T_0},\notag
\\&\hspace{10em}\forall\,\tilde{h}_{bm}^{(s)}\in\mathcal{H}_{bm}^{(s)},\label{Eq: Empirical distribution}
\\ \hat{\Pr}(\tilde{\varsigma};n+1)&\textstyle=\frac{a\hat{\Pr}(\tilde{\varsigma};n)}{n+1}+\frac{\mathbbm{1}\{\tilde{\varsigma}=\varsigma(n)\}}{n+1},~\forall\,\tilde{\varsigma}\in\mathcal{G}.\label{Eq: Empirical penalty distribution}
\end{align}
At the beginning of the next frame $n+1,$ \eqref{Eq: Empirical data mean} is uploaded for both statistics and state realization-based resource allocation scenarios whereas \eqref{Eq: Empirical distribution} and \eqref{Eq: Empirical penalty distribution} are only for the statistics-based resource allocation.

\begin{figure*}[t]
\centering
\subfigure[Statistics-based resource allocation.]
	{\def\svgwidth{\columnwidth}
		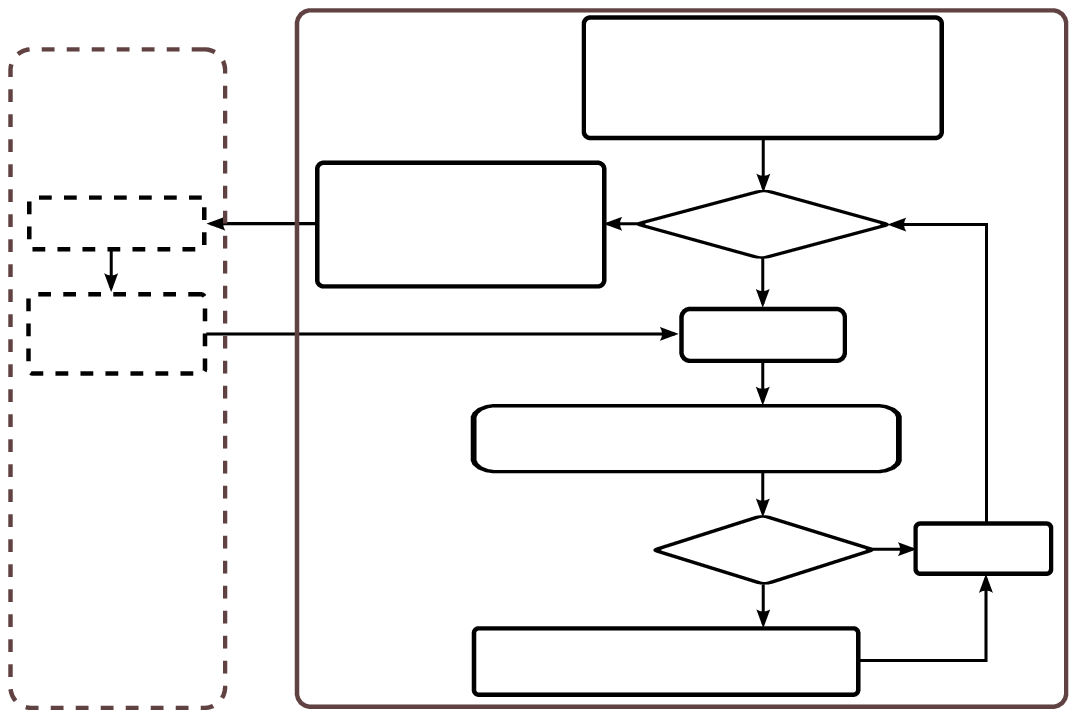}
\subfigure[State realization-based resource allocation.]
	{\def\svgwidth{\columnwidth}
		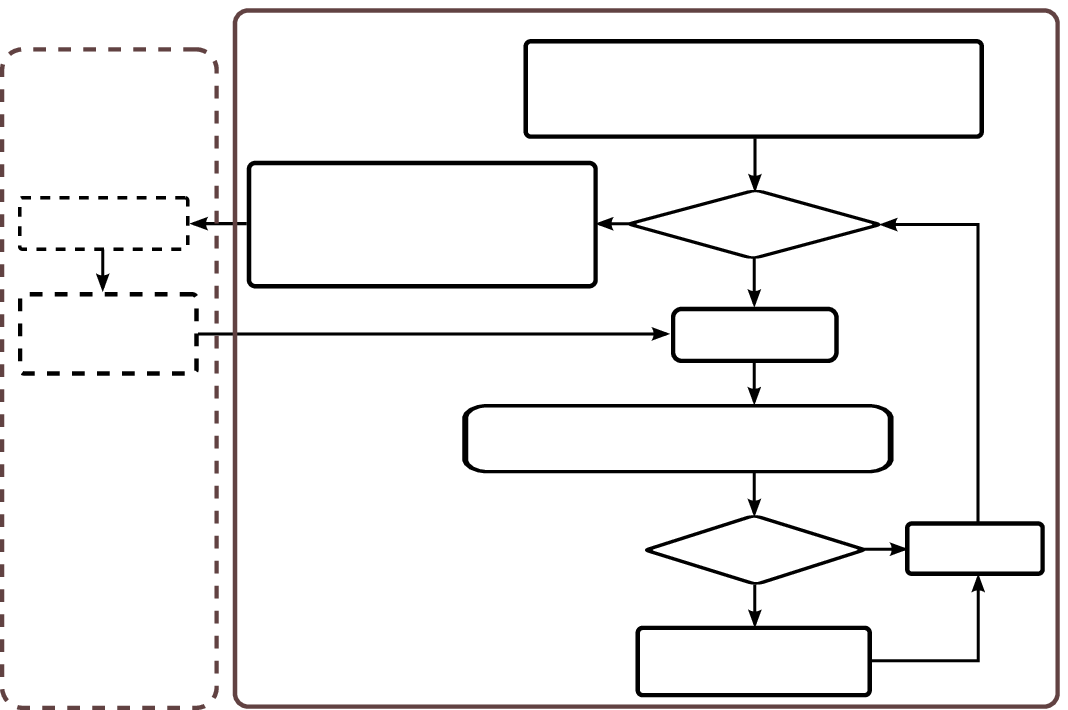}
\caption{Information flow diagrams of the two-timescale software-defined control scenarios.}
\label{Fig: Flow diagram}
\end{figure*}

\begin{remark}
Our proposed  software-defined control scenario operates in two timescales. In the long timescale,  BS $b$ uploads  $\hat{\lambda}_b$ and $\hat{\Pr}(\boldsymbol{\omega}_b$) (resp.~$\{\mathbf{h}_b(t)\}$) for the statistics-based scenario (resp.~state realization-based scenario) every $T_0$ slots. Based on the uploaded information, the controller finds the optimal CCE mapping rules $\{\Psi^{*}_{\rm I}(t)\}$  (resp.~$\Psi^{*}_{\rm II}$) and sends them back to the BSs. In the short timescale, BS $b$ schedules its UEs by solving \eqref{Eq: aided-bs-Lya-obj} in each time slot $t.$
\end{remark}
The information flow diagrams of our proposed two-timescale software-defined control scenarios are shown in Fig.~\ref{Fig: Flow diagram}, in which Fig.~\ref{Fig: Flow diagram}(a) describes the statistics-based approach whereas Fig.~\ref{Fig: Flow diagram}(b) shows the realization-based approach.
\section{Simulation Results}\label{Sec: results}

We evaluate the performance of the proposed software-defined control scenarios. In the simulation, we first consider an indoor office environment with two locally-coupled BSs and two UEs in each cell. For notational simplicity, the first BS and its served UEs  are referred to as BS 1, UE 1, and UE 2. Analogously, we have BS 2, UE 3, and UE 4 for the other cell. For UE 1 and UE 3, the distances to the serving and interfering BSs are 10\,m and 40\,m, respectively. The corresponding distances for UE 2 and UE 4 are 20\,m and 30\,m. Additionally, we consider the 2.4\,GHz carrier frequency with the path loss model $30 \log d +20 \log 2.4+46$~(dB) from \cite{rpt:itu_indoor}, where $d$ in meters is the distance between network entities. All channels experience Rayleigh fading with unit variance. We further assume that the channel gain and the time cost in fronthaul are quantized into two levels. Moreover, the fronthaul SNR measured at the controller is 20\,dB.  The coherence time and each sub-channel bandwidth are 100\,ms and 10\,MHz, respectively. Data arrivals are governed by Poisson processes with mean values $\bar{\lambda}_{11}=\bar{\lambda}_{12}=8$\,Mbps and $\bar{\lambda}_{23}=\bar{\lambda}_{24}=5$\,Mbps.  The other parameters are as follows:  $P_{b}=20$\,dBm, $\forall\,b\in\mathcal{B},$ $P_{\rm C}=25$\,dBm,  $|\mathcal{S}|=2,$ $T_0=10,$ $\mathcal{G} = \{0.25, 0.5\},$ $\sigma^2=-85$\,dBm, $R_{\rm{unit}} = 0.025\log_2 1.05$~(bit/s/Hz), $\kappa=10^4,$ and $V=100.$ To validate the advantage of our two proposed SDN resource allocation approaches, we consider a non-SDN baseline for comparison in which no central controller is deployed, and BSs do not communicate with each other. In the non-SDN baseline, the BS solves problem \eqref{Eq: aided-bs-Lya-obj} with $\mathcal{X}_b(t)= \mathcal{S},$ $\forall\,t\in\mathbb{Z}^{+},b\in\mathcal{B},$ and $\varsigma=0.$

\begin{figure}[t]
\centering
	\includegraphics[width=\columnwidth]{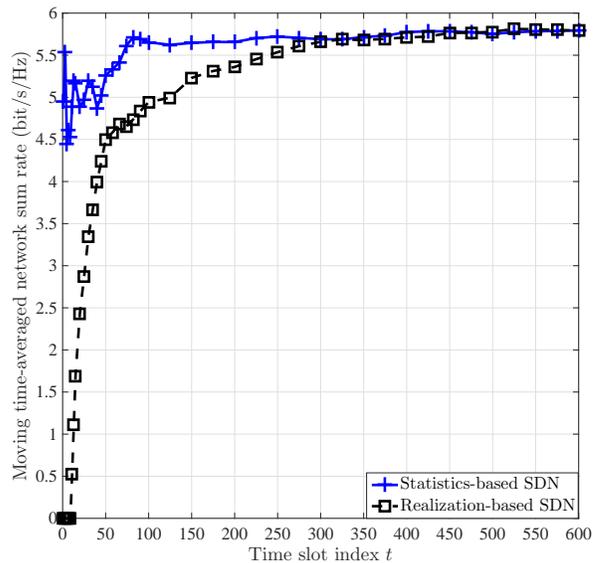}
	\caption{Convergence speed of the proposed controller-aided resource allocation schemes with respect to the moving average network sum rate.}
		\label{Fig:3}
\end{figure}
We first show how the performance of our two proposed controller-aided resource allocation approaches evolves with time. Since all BSs' DL rates and QoS requirements in terms of rates are taken into account in the controller's objective, we focus on the moving time-averaged network sum rate which is defined as $\frac{1}{t}\sum_{\tau=1}^{t}\sum_{b\in\mathcal{B}}\sum_{m\in\mathcal{M}_b}\sum_{s\in\mathcal{S}}R_{bm}^{(s)}(\tau).$ As shown in Fig.~\ref{Fig:3}, the moving average rate of the statistics-based SDN scheme converges faster than the moving average rate of the state realization-based SDN scheme. Specifically, the statistics-based  approach achieves steady-state performance at $t=100$ whereas the convergence of the realization-based  approach is three times slower.  The statistics-based SDN approach has  faster convergence speed  since it leverages the estimated channel distribution which includes more information than the channel realization used by  the realization-based SDN approach. However, owing to the distribution being empirically estimated, the statistics-based approach still needs extra time slots to gather full statistical information to achieve convergence. 

	\begin{figure}[t]
\centering
		\includegraphics[width=\columnwidth]{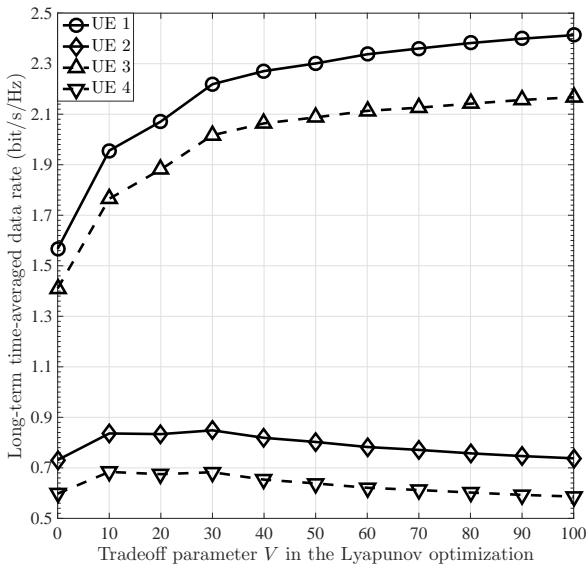}
	\caption{UE's long-term time-averaged data rate versus the Lyapunov optimization tradeoff factor  $V.$}
	\label{Fig:4}
\end{figure}

In Fig.~\ref{Fig:4} and Fig.~\ref{Fig:5}, we  investigate the UE's long-term time-averaged data rate and queue length as the tradeoff control parameter $V$  varies. Before discussing these two figures, we note that the tradeoff between the average rate/queue length and $V$ is similar in the two proposed SDN approaches. Since the performance of the realization-based approach is better due to lighter fronthaul overhead, we consider the realization-based SDN approach to illustrate the tradeoff.\footnote{The performance of the two proposed SDN approaches will be compared in Figs.~\ref{Fig:6}--\ref{Fig:9}.}  When $V=0,$  the BSs focus on queue stabilization. As illustrated in Remark \ref{Remark: V impact}, the shortest queue length is achieved despite UEs' channel quality, which results in the lowest data rate.  As we increase $V,$ the BS focuses more on sum rate maximization. In order to improve the sum rate, the BS allocates  more resources to the UEs that have better channel quality, i.e., UE 1 and UE 3. Hence, in Fig.~\ref{Fig:4}, it can be observed  that UE 1's and UE 3's average rates are further improved as $V$ increases. On the other hand, the average queue lengths of UE 2 and UE 4 grow monotonically with $V.$ Moreover, UE 1 and UE 2 have higher data rates since their larger traffic demands are taken into account in the controller's objective.
Next, we compare the performance of the two proposed software-defined control schemes and the non-SDN scheme. Versus the tradeoff parameter $V,$ we show each BS's long-term time-averaged DL rate and queue length of the three studied schemes in  Fig.~\ref{Fig:6} and Fig.~\ref{Fig:7}, respectively. It  can be observed that the BS's average DL rate and queue length increase with $V$  as per Remark \ref{Remark: V impact}. Additionally, with the aid of the controller's recommendations, our proposed approaches achieve better throughput than the non-SDN scheme. Furthermore,  Fig.~\ref{Fig:6}  illustrates that the proposed SDN schemes achieve higher DL rates at higher traffic demands, i.e. traffic-aware control mechanisms, whereas the  non-SDN scheme is agnostic to traffic demands.  The controller's traffic-aware resource allocation also affects the BSs' average queue length. This is shown in  Fig.~\ref{Fig:7} where BS 1's queue length in the proposed SDN schemes is smaller than in the non-SDN scheme for all $V$ values.
\begin{figure}[t]
\centering
	\includegraphics[width=\columnwidth]{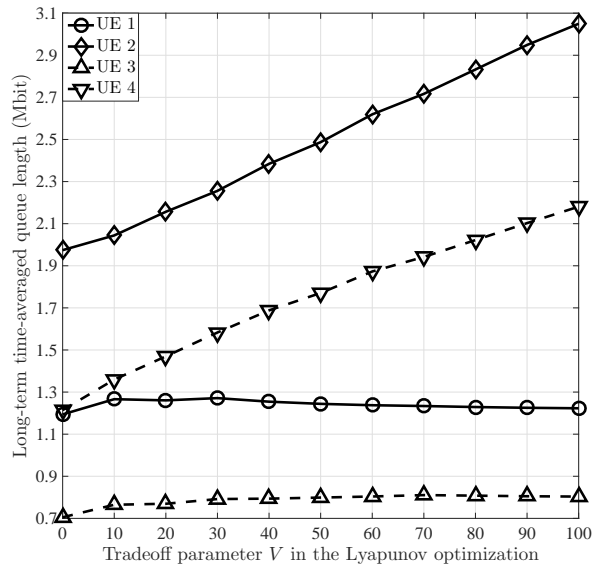}
	\caption{UE's long-term time-averaged queue length versus the Lyapunov optimization tradeoff factor  $V.$}
		\label{Fig:5}
		\end{figure}
\begin{figure}
\centering
		\includegraphics[width=\columnwidth]{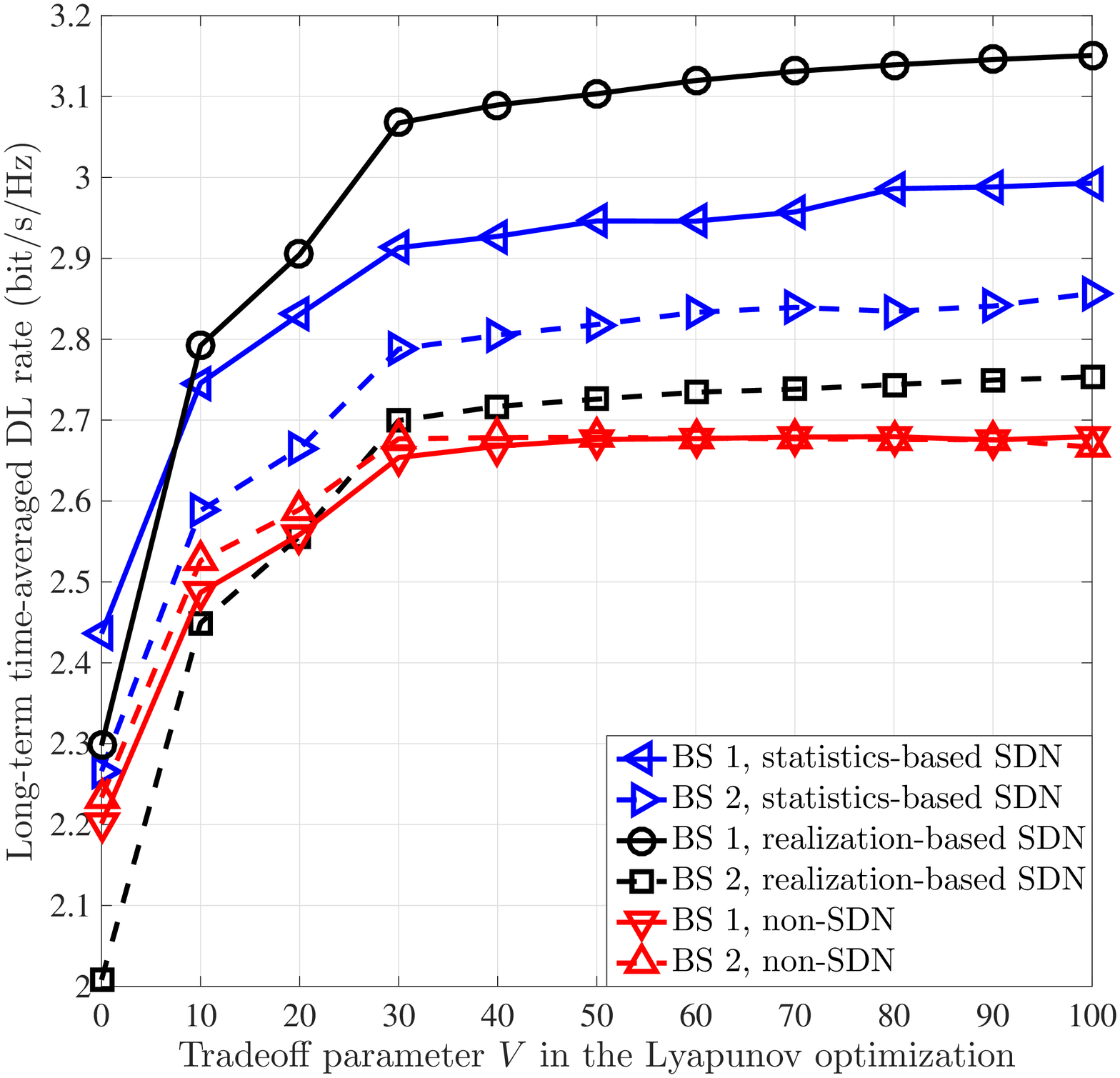}
	\caption{BS's long-term time-averaged DL rate versus the Lyapunov optimization tradeoff factor $V.$}
	\label{Fig:6}
\end{figure}
\begin{figure}[t]
\centering
	\includegraphics[width=\columnwidth]{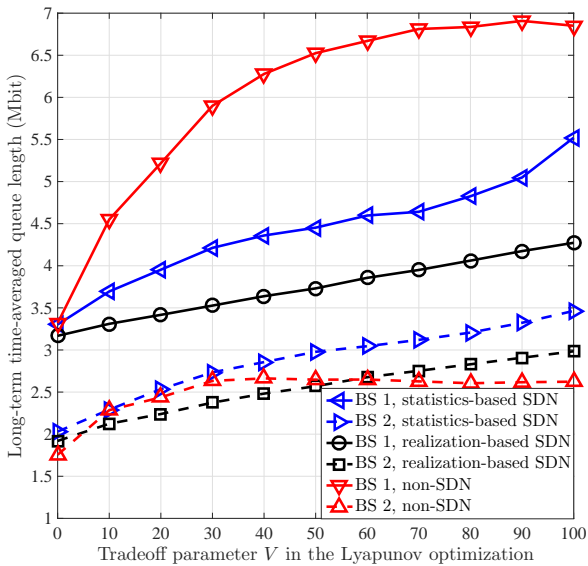}
		\caption{BS's long-term time-averaged queue length versus the Lyapunov optimization tradeoff factor  $V.$}
		\label{Fig:7}
\end{figure}

Fig.~\ref{Fig:8} illustrates the tradeoff between the long-term time-averaged network sum rate and queue length. Given an average network throughput, our proposed SDN approaches have lower average queue lengths than the non-SDN scheme.  In other words, our approaches achieve better delay performance since latency is proportional to the average queue length \cite{Littlelaw}. In this regard, the realization-based SDN approach has  40\% latency reduction over the non-SDN scheme while up to  34\% latency  reduction is achieved by the statistics-based SDN scheme. Moreover, we note that in contrast with the non-SDN baseline, the proposed SDN approaches  achieve higher network throughput while guaranteeing better network latency performance.

\begin{figure}
\centering
		\includegraphics[width=\columnwidth]{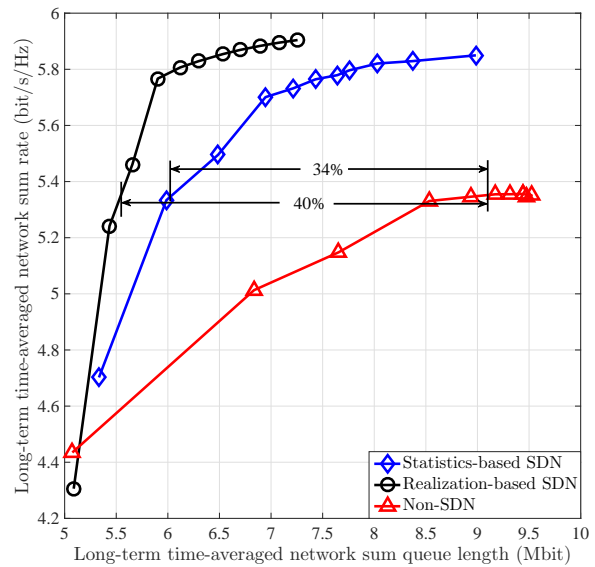}
	\caption{Tradeoff between the long-term time-averaged network sum rate and queue length.}
	\label{Fig:8}
\end{figure}

The impact of fronthaul reliability on the network average throughput and latency performance of the software-defined wireless network is shown in Fig.~\ref{Fig:9}. Since there is no fronthaul in the non-SDN scheme, its performance does not vary with the fronthaul SNR as illustrated  in Fig.~\ref{Fig:9}. For low fronthaul SNR values, the BS in the SDN schemes requires more than  $[\mathcal{G}]_{\rm max}$ time portion, i.e., the maximum element in the set $\mathcal{G},$ to upload its local information to acquire recommendations. Since the controller's recommendations are not available after  $[\mathcal{G}]_{\rm max}$  time portion,  the BS utilizes all sub-carriers as in the non-SDN scheme to schedule the UEs. However, due to the allocated time $[\mathcal{G}]_{\rm max}$ in the fronthaul, the SDN schemes have lower average throughput in contrast with the non-SDN scheme for low fronthaul SNR values. As the fronthaul becomes reliable in terms of increased fronthaul SNR, the throughput gains provided by the controller's recommendations dominate the overhead in the fronthaul. The improved throughput effectively decreases the average queue length. Note that the performance enhancement is more prominent when the fronthaul is more reliable. Furthermore, owing to the lighter burden on the fronthaul, the realization-based SDN scheme achieves  better throughput and latency performance than the statistics-based SDN scheme.

\begin{figure}[t]
\centering
	\includegraphics[width=\columnwidth]{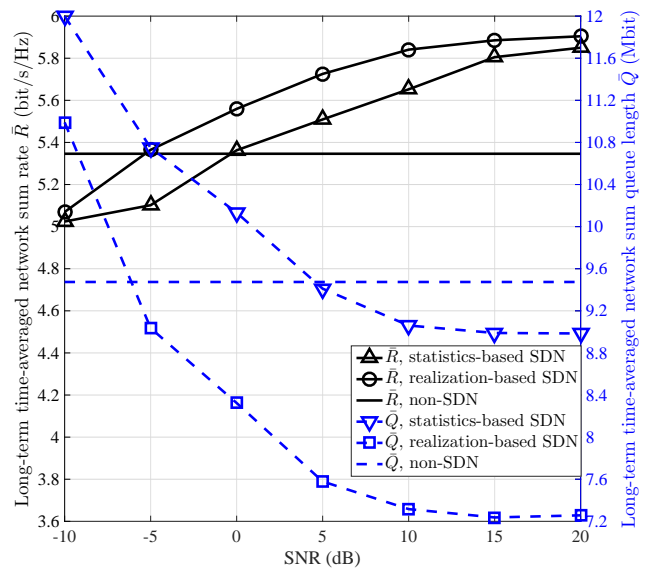}
		\caption{Long-term time-averaged network sum rate and queue length as the fronthaul SNR varies.}
	\label{Fig:9}
	\end{figure}
	
\begin{figure}
\centering
	\includegraphics[width=\columnwidth]{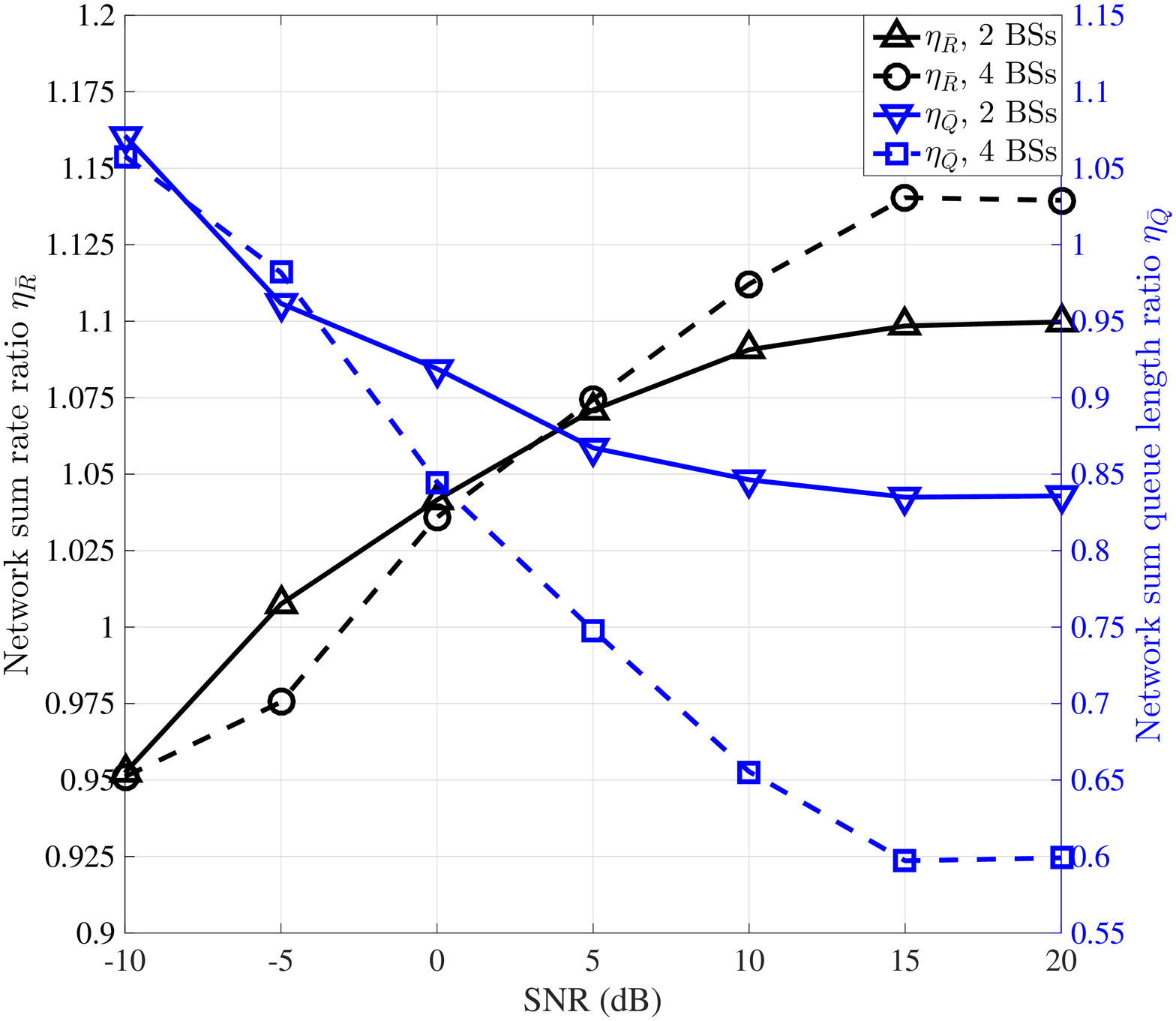}
	\caption{Network sum rate and queue length ratios as the fronthaul SNR varies for different BS densities.}
	\label{Fig:10}
\end{figure}

Finally, Fig.~\ref{Fig:10} shows the benefit of the controller's recommendations when the number of locally-coupled BSs increases. In this comparison with the non-SDN architecture, we consider the realization-based SDN approach due to its higher throughput gain and latency reduction which are manifested in Fig.~\ref{Fig:8} and Fig.~\ref{Fig:9}. Regarding the network setting in Fig.~\ref{Fig:10}, we assume that each BS serves two UEs. One of the two UEs is closer to the serving BS with the 10\,m distance and  the 40\,m distance to all interfering BSs. The distances from the other UE to its serving BS and all interfering distances are 20\,m and 30\,m, respectively. In addition, $\bar{\lambda}_{bm}=7\,\mbox{Mbps},\forall\,b\in\mathcal{B},m\in\mathcal{M}_b.$ 
In order to provide a clear comparison for the rate and latency enhancements, the throughputs and queue lengths of the proposed scheme are normalized by the corresponding rates and queue lengths of the non-SDN scheme.
 In this regard, we denote the metrics  for the throughput  ratio and the queue length ratio as $\eta_{\bar{R}}$ and $\eta_{\bar{Q}},$ respectively. 
When the fronthaul SNR is low, both schemes are unable to obtain the recommendations after spending $[\mathcal{G}]_{\rm max}$ time portion in fronthaul. Thus,  both cases have the same $\eta_{\bar{R}},$ i.e., $(T_0-[\mathcal{G}]_{\rm max})/T_0,$ which results in the same  $\eta_{\bar{Q}}.$
Note that the fronthaul overhead is proportional to the number of BSs as per \eqref{Eq: FH time cost} and \eqref{Eq: FH download cost}. As the fronthaul SNR increases, the two-BS case achieves more throughput gains with less time cost. 
Nevertheless, when the performance gain dominates at a high SNR, the  throughput and latency enhancements brought by the controller recommendations are more prominent for $|\mathcal{B}|=4$. Higher performance gain is achieved since more coupled interference is coordinated by the controller.
In other words, given that the controller's recommendations can be acquired in a timely and reliable manner, the software-defined control scenario is more beneficial when the BS density grows.

\section{Conclusions}\label{Sec: conclusion}
This work has investigated a software-defined control mechanism for RANs taking into account network dynamics and capacity-limited  fronthaul links. 
In the considered network, BSs compete for resources to maximize their DL rates while stabilizing the queue lengths. Due to queue dynamics and the randomness of the wireless channel, BSs' rate maximization problems are cast as dynamic stochastic games that are coordinated by an SDN controller via an in-band wireless fronthaul.
Here, taking into account a time penalty in the fronthaul, we have proposed a two-timescale software-defined control mechanism in which the controller issues CCE recommendations in the slow timescale. To this end,  we have considered two approaches, which trade off the utility convergence speed and the burden on the fronthaul, to find the optimal CCE.
Subsequently, incentivized by the CCE recommendations, i.e., allocated sub-carriers, the BS leverages Lyapunov optimization techniques to schedule its UEs in a low-complexity latency-aware manner at each time slot. 
Numerical results have shown that our proposed SDN approaches simultaneously achieve higher throughput and guarantee lower  latency over the non-SDN baseline. Specifically, up to 40\% latency is reduced with the help of the SDN controller. When the fronthaul becomes reliable, performance enhancement makes the overhead of acquiring controller's recommendations negligible. Moreover, the controller's recommendations are more beneficial as the BS density increases.

\appendices

\section{Proof of Proposition \ref{Prop: Mean rate stability}}
\label{Lem: Mean rate stability}

As the maximum interference channel gain $[\mathcal{H}_{b'm}^{(s)}]_{\rm max}$ is considered in $v_b(\boldsymbol{\omega},\boldsymbol{\alpha}),$ we can straightforwardly find
\begin{equation}\label{Eq: Utility bound}
u_{b}(\boldsymbol{\omega},\boldsymbol{\alpha})\geq v_{b}(\boldsymbol{\omega},\boldsymbol{\alpha}),~\forall\,\boldsymbol{\omega}\in\mathcal{W},\boldsymbol{\alpha}\in\mathcal{A}.
\end{equation}
Taking the expectation of \eqref{Eq: Utility bound} with respect to $\Pr(  \boldsymbol{\omega})\Pr(\boldsymbol{\alpha}|  \boldsymbol{\omega})$ and applying \eqref{Eq: BS's utility} and \eqref{Eq: Auxiliary mean stable}, we obtain
$ \sum_{m\in\mathcal{M}_b}\sum_{s\in\mathcal{S}}\bar{R}_{bm}^{(s)}\geq	\bar{\lambda}_{b}.$

 \section{Proof of Proposition \ref{Prop: epsilon-CCE}}
\label{Lem: epsilon-CCE}

Considering the auxiliary utility $v_b(\boldsymbol{\omega},\boldsymbol{\alpha})$ and the CCE strategy ${\Pr}_{v}(\boldsymbol{\alpha}|  \boldsymbol{\omega})$  with respect to  $v_b(\boldsymbol{\omega},\boldsymbol{\alpha})$, we first rewrite  \eqref{Eq: epsilon-CCE-1} and \eqref{Eq: epsilon-CCE-2} as
\begin{align}
\Pr(\boldsymbol{\omega}_b)\theta_b( \boldsymbol{\omega}_b) &\textstyle\geq  \sum\limits_{\boldsymbol{\omega}\in\mathcal{W}|\boldsymbol{\omega}_b}\sum\limits_{\boldsymbol{\alpha}\in\mathcal{A}}  \Pr(\boldsymbol{\omega}) {\Pr}_{v}(\boldsymbol{\alpha}|  \boldsymbol{\omega}) v_b(\boldsymbol{\omega},\boldsymbol{\chi}_b, \boldsymbol{\alpha}_{-b}), \notag
\\&  \hspace{4.8em}~\forall\,b\in\mathcal{B},\boldsymbol{\omega}_b\in\mathcal{W}_b, \boldsymbol{\chi}_b\in\mathcal{A}_b,\label{Eq: Auxiliary CCE-1}
\\\bar{v}_b&\textstyle\geq
 \sum\limits_{ \boldsymbol{\omega}_b\in\mathcal{W}_b} \Pr( \boldsymbol{\omega}_b)\theta_b( \boldsymbol{\omega}_b),~\forall\,b\in\mathcal{B}.\label{Eq: Auxiliary CCE-2}
\end{align}
Note that given the observed state $\boldsymbol{\omega}_b$, $\theta_b( \boldsymbol{\omega}_b)$ is the maximum utility if BS $b$ deviates from the CCE strategy. Thus, we can express $\theta_b( \boldsymbol{\omega}_b)$ as
\begin{align}
\theta_b( \boldsymbol{\omega}_b)= \max\limits_{\boldsymbol{\chi}_b\in\mathcal{A}_b}\bigg\{ &\textstyle\sum\limits_{\boldsymbol{\omega}\in\mathcal{W}|\boldsymbol{\omega}_b}\sum\limits_{\boldsymbol{\alpha}\in\mathcal{A}}   \frac{\Pr(\boldsymbol{\omega})}{\Pr(\boldsymbol{\omega}_b)}\notag
\\&\times{\Pr}_{v}(\boldsymbol{\alpha}|  \boldsymbol{\omega})v_{b}(\boldsymbol{\omega},\boldsymbol{\chi}_{b},\boldsymbol{\alpha}_{-b})\Big\}.\label{Eq: Auxiliary CCE-max}
\end{align}
Analogous to \eqref{Eq: Auxiliary CCE-max}, we define $\mu_b( \boldsymbol{\omega}_b)$,  for the utility $u_b(\boldsymbol{\omega},\boldsymbol{\alpha})$, as
\begin{align}
\mu_b( \boldsymbol{\omega}_b)&\textstyle\coloneqq\max\limits_{\boldsymbol{\chi}_b\in\mathcal{A}_b}\bigg\{ \sum\limits_{\boldsymbol{\omega}\in\mathcal{W}|\boldsymbol{\omega}_b}\sum\limits_{\boldsymbol{\alpha}\in\mathcal{A}}   \frac{\Pr(\boldsymbol{\omega})}{\Pr(\boldsymbol{\omega}_b)}{\Pr}_{v}(\boldsymbol{\alpha}|  \boldsymbol{\omega})\notag
\\&\hspace{3em} u_{b}(\boldsymbol{\omega},\boldsymbol{\chi}_{b},\boldsymbol{\alpha}_{-b})\Big\},~\forall\,b\in\mathcal{B},\boldsymbol{\omega}_b\in\mathcal{W}_b,\label{Eq: Auxiliary epsilon-CCE-max}
\end{align}
which can be further rewritten  as
\begin{align}
 &\textstyle\Pr(\boldsymbol{\omega}_b)\mu_b( \boldsymbol{\omega}_b)\geq  \sum\limits_{\boldsymbol{\omega}\in\mathcal{W}|\boldsymbol{\omega}_b}\sum\limits_{\boldsymbol{\alpha}\in\mathcal{A}}  \Pr(\boldsymbol{\omega}) {\Pr}_{v}(\boldsymbol{\alpha}|  \boldsymbol{\omega}) \notag
\\&\qquad\textstyle\times   u_b(\boldsymbol{\omega},\boldsymbol{\chi}_b, \boldsymbol{\alpha}_{-b}),~\forall\,b\in\mathcal{B},\boldsymbol{\omega}_b\in\mathcal{W}_b, \boldsymbol{\chi}_b\in\mathcal{A}_b.\label{Eq: Epsilon-CCE proo1}
\end{align}
Additionally, we rewrite \eqref{Eq: Utility bound} as
\begin{align}
u_{b}(\boldsymbol{\omega},\boldsymbol{\alpha})&= v_{b}(\boldsymbol{\omega},\boldsymbol{\alpha})+\delta_{b}(\boldsymbol{\omega},\boldsymbol{\alpha}),\label{Eq: delta}
\end{align}
where $\delta_{b}(\boldsymbol{\omega},\boldsymbol{\alpha})\geq 0.$
Then, using \eqref{Eq: delta} in \eqref{Eq: Auxiliary epsilon-CCE-max}, we obtain
\begin{align}
\mu_b( \boldsymbol{\omega}_b)&\textstyle=\max\limits_{\boldsymbol{\chi}_b\in\mathcal{A}_b}\bigg\{ \sum\limits_{\boldsymbol{\omega}\in\mathcal{W}|\boldsymbol{\omega}_b}\sum\limits_{\boldsymbol{\alpha}\in\mathcal{A}}   \frac{\Pr(\boldsymbol{\omega})}{\Pr(\boldsymbol{\omega}_b)}{\Pr}_{v}(\boldsymbol{\alpha}|  \boldsymbol{\omega})\notag
\\&\qquad\textstyle\times\big[v_{b}(\boldsymbol{\omega},\boldsymbol{\chi}_{b},\boldsymbol{\alpha}_{-b})+\delta_{b}(\boldsymbol{\omega},\boldsymbol{\chi}_{b},\boldsymbol{\alpha}_{-b})\big]\Big\}\notag
\\&\textstyle\geq \max\limits_{\boldsymbol{\chi}_b\in\mathcal{A}_b}\bigg\{ \sum\limits_{\boldsymbol{\omega}\in\mathcal{W}|\boldsymbol{\omega}_b}\sum\limits_{\boldsymbol{\alpha}\in\mathcal{A}}   \frac{\Pr(\boldsymbol{\omega})}{\Pr(\boldsymbol{\omega}_b)}{\Pr}_{v}(\boldsymbol{\alpha}|  \boldsymbol{\omega})\notag
\\&\qquad\textstyle\times v_{b}(\boldsymbol{\omega},\boldsymbol{\chi}_{b},\boldsymbol{\alpha}_{-b})\big\}=\theta_b( \boldsymbol{\omega}_b),\label{Eq: appendix-1}
\end{align}
The inequality  in \eqref{Eq: appendix-1} is obtained due to $\delta_b(\boldsymbol{\omega},\boldsymbol{\alpha})\geq 0$.
We can further express $\mu_b( \boldsymbol{\omega}_b)=\theta_b( \boldsymbol{\omega}_b)+ \epsilon_b( \boldsymbol{\omega}_b)$ with
\begin{multline}
\textstyle \epsilon_b( \boldsymbol{\omega}_b)=-\theta_b( \boldsymbol{\omega}_b)+\max\limits_{\boldsymbol{\chi}_b\in\mathcal{A}_b}\bigg\{
\sum\limits_{\boldsymbol{\omega}\in\mathcal{W}|\boldsymbol{\omega}_b}\sum\limits_{\boldsymbol{\alpha}\in\mathcal{A}}    \frac{\Pr(\boldsymbol{\omega})}{\Pr(\boldsymbol{\omega}_b)} {\Pr}_{v}(\boldsymbol{\alpha}|  \boldsymbol{\omega})
\\\textstyle\times\big[v_b(\boldsymbol{\omega},\boldsymbol{\chi}_{b},\boldsymbol{\alpha}_{-b})
+\delta_b(\boldsymbol{\omega},\boldsymbol{\chi}_{b},\boldsymbol{\alpha}_{-b})\big]\Big\}\geq 0.
\end{multline}
Subsequently, applying  \eqref {Eq: Utility bound} and $\mu_b( \boldsymbol{\omega}_b)=\theta_b( \boldsymbol{\omega}_b)+\epsilon_{b}( \boldsymbol{\omega}_b)$   to \eqref{Eq: Auxiliary CCE-2}, we obtain, $\forall\,b\in\mathcal{B}$,
\begin{multline}
 \textstyle\sum\limits_{\boldsymbol{\omega}\in\mathcal{W}}\sum\limits_{\boldsymbol{\alpha}\in\mathcal{A}}  \Pr(\boldsymbol{\omega}){\Pr}_{v}(\boldsymbol{\alpha}|  \boldsymbol{\omega}) u_{b}(\boldsymbol{\omega},\boldsymbol{\alpha})
 \geq\sum\limits_{\boldsymbol{\omega}_b\in\mathcal{W}_b} \Pr(\boldsymbol{\omega}_{b}) \theta_b( \boldsymbol{\omega}_b)
 \\\geq\sum\limits_{\boldsymbol{\omega}_b\in\mathcal{W}_b}\Pr(\boldsymbol{\omega}_b)\mu_b( \boldsymbol{\omega}_b)-\epsilon,\label{Eq: Epsilon-CCE proof 2}
\end{multline}
where $\epsilon=\max_{b\in\mathcal{B}}\{\sum_{\boldsymbol{\omega}_b\in\mathcal{W}_b}\Pr(\boldsymbol{\omega}_b)\epsilon_b( \boldsymbol{\omega}_b)\}\geq 0.$ Note that ${\Pr}_{v}(\boldsymbol{\alpha}|  \boldsymbol{\omega})$ achieves the $\epsilon$-CCE with respect to $u_{b}(\boldsymbol{\omega},\boldsymbol{\alpha})$ from \eqref{Eq: Epsilon-CCE proo1} and \eqref{Eq: Epsilon-CCE proof 2}.

 \section{Proof of Lemma \ref{lemma power}}
\label{Lem: opt pow}

Since the logarithmic  function in (40a) monotonically increases with $P_{bm}^{(s)}$, we can infer that the optimum of  (40) occurs when $\sum_{m\in\mathcal{M}_b}\sum_{s\in\mathcal{X}_b}P_{bm}^{(s)} = |\mathcal{S}|\totalpower{b}$.
Thus, problem (40) is equivalent to
\begin{IEEEeqnarray*}{cl}
 \underset{ P_{bm}^{(s)}}{\mbox{maximize}}&\textstyle ~~\sum\limits_{m\in\mathcal{M}_b}\sum\limits_{s\in\mathcal{X}_b}(Q_{bm}+V)
\\&\textstyle ~~\qquad\times\mathbb{E}_{I_{bm}^{(s)}}  \bigg[\ln\bigg(1+\frac{P_{bm}^{(s)}h_{bm}^{(s)}}{\sigma^2+ I_{bm}^{(s)}} \bigg)\Big|\boldsymbol{\omega}_b,\mathcal{X}_b\bigg]
\\\mbox{subject to}& \textstyle~~\sum\limits_{m\in\mathcal{M}_b}\sum\limits_{s\in\mathcal{X}_b}P_{bm}^{(s)} = |\mathcal{S}|\totalpower{b},
\\&~~ P_{bm}^{(s)}  \geq 0,~\forall\,m\in\mathcal{M}_b,s\in\mathcal{X}_b,
\end{IEEEeqnarray*}
which is a  convex optimization problem. Subsequently, applying the KKT conditions, the optimal solution $P_{bm}^{(s)*},\forall\,m\in\mathcal{M}_b,s\in\mathcal{X}_b,$ satisfies
\begin{subequations}\label{Eq: KKT}
\begin{numcases}
\textstyle \mathbb{E}_{I_{bm}^{(s)}}  \bigg[\frac{   \big(Q_{bm}+V\big)h_{bm}^{(s) }}  {  \sigma^2+ I_{bm}^{(s)} +P_{bm}^{(s)*}h_{bm}^{(s)}  } \bigg|\boldsymbol{\omega}_b,\mathcal{X}_b\bigg]\notag
\\\hspace{5em}=\gamma_b-\gamma_{bm}^{(s)} ,~\forall\,m\in\mathcal{M}_b,s\in\mathcal{X}_b,\label{Eq: KKT-1}
\\ P_{bm}^{(s)*}  \geq 0,~\forall\,m\in\mathcal{M}_b,s\in\mathcal{X}_b \label{Eq: KKT-2}
\\\gamma_{bm}^{(s)}  \geq 0,~\forall\,m\in\mathcal{M}_b,s\in\mathcal{X}_b\label{Eq: KKT-3}
\\P_{bm}^{(s)*}\gamma_{bm}^{(s)}   = 0 ,~\forall\,m\in\mathcal{M}_b,s\in\mathcal{X}_b,\label{Eq: KKT-4}
\\\textstyle\sum\limits_{m\in\mathcal{M}_b}\sum\limits_{s\in\mathcal{X}_b}P_{bm}^{(s)*} = |\mathcal{S}|\totalpower{b},\label{Eq: KKT-5}
\end{numcases}
\end{subequations}
where $\gamma_b\in\mathbb{R}$  and $\gamma_{bm}^{(s)},\forall\,m\in\mathcal{M}_b,s\in\mathcal{X}_b,$ are the Lagrange multipliers. 
From \eqref{Eq: KKT-2}, \eqref{Eq: KKT-3}, and \eqref{Eq: KKT-4}, we deduce that $P_{bm}^{(s)*}  = 0$ if  $\gamma_{bm}^{(s)}   > 0$. Additionally, $ \gamma_{bm}^{(s)}   = 0$ when $P_{bm}^{(s)*}> 0$.
 Therefore,  if $\mathbb{E}_{I_{bm}^{(s)}}\Big[\frac{(Q_{bm}+V)h_{bm}^{(s)}}{\sigma^2+ I_{bm}^{(s)}}\Big|\boldsymbol{\omega}_b,\mathcal{X}_b \Big]>\gamma_b$, we have  $ \gamma_{bm}^{(s)}   = 0$ and find the optimal power $P_{bm}^{(s)*}> 0$ such that
\begin{align}
\textstyle\mathbb{E}_{I_{bm}^{(s)}}\bigg[\frac{\big(Q_{bm}+V\big)h_{bm}^{(s)}}{\sigma^2+ I_{bm}^{(s)}+P_{bm}^{(s)*}h_{bm}^{(s)}}
 \bigg|\boldsymbol{\omega}_b,\mathcal{X}_b\bigg] =\gamma_b.\label{Eq: KKT water-willing}
\end{align}
Otherwise, $P_{bm}^{(s)*}=0$, and we select a $ \gamma_{bm}^{(s)}   \geq 0$   such that  $\mathbb{E}_{I_{bm}^{(s)}}\Big[\frac{(Q_{bm}+V)h_{bm}^{(s)}}{\sigma^2+ I_{bm}^{(s)}}\Big|\boldsymbol{\omega}_b,\mathcal{X}_b \Big]=\gamma_b- \gamma_{bm}^{(s)} $.  Moreover, $\gamma_b$ affects the power values as per \eqref{Eq: KKT water-willing}. Since the allocated power has to satisfy \eqref{Eq: KKT-5},  $\gamma_b$ is chosen such that $\sum_{m\in\mathcal{M}_b}\sum_{s\in\mathcal{X}_b}P_{bm}^{(s)*} =|\mathcal{S}|\totalpower{b}.$

\bibliographystyle{IEEEtran}
\bibliography{ref}

\end{document}